%% file: main.tex
\PassOptionsToPackage{dvipsnames,table}{xcolor}
\documentclass[11pt,letterpaper]{article}
\input{mancro}
\graphicspath{{Draft/}{./}}

\newtheorem{theorem}{Theorem}[section]
\newtheorem{lemma}[theorem]{Lemma}
\newtheorem{corollary}[theorem]{Corollary}
\newtheorem{proposition}[theorem]{Proposition}
\theoremstyle{remark}
\newtheorem{remark}[theorem]{Remark}
\numberwithin{equation}{section}
\DeclareMathOperator{\Tr}{Tr}
\DeclareMathOperator{\rank}{rank}
\DeclareMathOperator{\dist}{dist}
\newcommand{\DD}{\mathcal D}
\newcommand{\ZZ}{\mathcal Z}
\newcommand{\joint}{\Psi}
\newcommand{\jointstar}{\Psi^*}
\newcommand{\gradient}{F}
\newcommand{\Nzero}{\mathbb N_0}
\newcommand{\norm}[1]{\lVert #1\rVert}
\renewcommand{\O}{\mathcal O}

\title{Average-and Last-Iterate Lower Bounds for Optimistic Matrix Mirror-Prox in Quantum Zero-Sum Games}
\author{Yiheng Su\,\thanks{Department of Computer Sciences, University of Wisconsin-Madison, email: \texttt{su228@wisc.edu}} \and Emmanouil-Vasileios Vlatakis-Gkaragkounis\,\thanks{Department of Computer Sciences, University of Wisconsin-Madison, email: \texttt{vlatakis@wisc.edu}}\and  Pucheng Xiong\,\thanks{Department of Computer Sciences, University of Wisconsin-Madison, email: \texttt{pxiong79@wisc.edu}}}
\date{}

\begin{document}
\maketitle
\begin{abstract}
Optimistic matrix mirror-prox (OMMP) computes $\varepsilon$-approximate Nash equilibria in
quantum zero-sum games with an $\mathcal{O}(1/\varepsilon)$ average-iterate
guarantee \cite{VasconcelosEtAl2025}. However, it is natural to ask whether this dependence on accuracy is tight, and whether geometric last-iterate convergence can be guaranteed. We study these questions through explicit games with
one qubit per player. First, we prove an $\Omega(1/\varepsilon)$ lower bound
for the uniform-average output that includes the maximally mixed initial
state, independently of the regularizer and step size. Second, we construct
a fixed game on which optimistic gradient descent--ascent (OGDA), initialized at the maximally mixed state, has last-iterate Frobenius distance to equilibrium \(\Theta(1/t)\) and duality gap \(\Theta(1/t^{3})\) for every sufficiently small fixed step size. A separate fixed game exhibits arbitrarily long delays in reducing the initial error by a constant factor across a family of initial states.
Finally,
we give a fixed game with a unique, strictly complementary equilibrium on
which optimistic matrix multiplicative weights updates (OMMWU) converge only polynomially (at the last iterate) from the maximally mixed state for every
fixed positive step size. 
The Frobenius distance and quantum relative entropy from the equilibrium to the iterates decay as \(\Theta(1/t)\), while the duality gap decays as \(\Theta(1/t^{2})\).
% The distance and quantum relative entropy decay as
% $\Theta(1/t)$, while the duality gap decays as $\Theta(1/t^{2})$.
\end{abstract}

{\small \tableofcontents}
\section{Introduction}
\label{sec:introduction}

Quantum zero-sum games extend classical matrix games by allowing each
player to choose a density matrix. The players send their states to a
referee, whose joint measurement determines opposite payoffs. The resulting
expected payoff is bilinear in the two states, and a Nash equilibrium is
a saddle point over their strategy spaces. These spaces are
\emph{spectraplexes}: sets of positive semidefinite matrices with unit trace.
Although this formulation retains the convexity of classical zero-sum
games, it changes the geometry of learning. A quantum state can change both
its eigenvalues and its eigenspaces, and the boundary contains a continuum
of pure states. Understanding how these features affect equilibrium
computation is the focus of this paper.

\subsection{Prior works and a motivating question}

Jain and Watrous~\cite{JainWatrous2009} developed a matrix
multiplicative weights method for approximating equilibria in
non-interactive quantum zero-sum games. For normalized payoffs and
$d$ qubits per player, its iteration complexity is
$\mathcal{O}(d/\varepsilon^2)$. Classical saddle-point optimization
suggests a route to improving this dependence. Nemirovski's
mirror-prox method~\cite{Nemirovski2004} achieves an
$\mathcal{O}(1/T)$ averaged gap for Lipschitz monotone problems,
while optimistic mirror descent uses past gradients to predict
future feedback~\cite{RakhlinSridharan2013}.
Vasconcelos et al.~\cite{VasconcelosEtAl2025} brought this approach
to quantum games through the optimistic matrix mirror-prox (OMMP)
framework. Its instantiation with negative von Neumann entropy,
optimistic matrix multiplicative weights updates (OMMWU), computes
an averaged $\varepsilon$-Nash equilibrium in
$\mathcal{O}(d/\varepsilon)$ iterations, using one new gradient
evaluation per round after initialization. They conjectured that
the $1/\varepsilon$ dependence is tight for OMMP.

The sharpness of an averaged guarantee is separate from the
convergence rate of the strategies actually played. This distinction
is particularly important because optimism can yield geometric
last-iterate convergence in classical matrix games.
Daskalakis and Panageas~\cite{DaskalakisPanageas2019} established
last-iterate convergence of optimistic multiplicative weights
updates (OMWU) under a unique equilibrium assumption.
Wei et al.~\cite{WeiEtAl2021} subsequently proved geometric
convergence under the same assumption with a constant step size,
giving an $\mathcal{O}(\log(1/\varepsilon))$ iteration bound.
They also proved geometric convergence of optimistic gradient
descent--ascent (OGDA) for bilinear games over polytopes, without
requiring uniqueness. Their analysis uses saddle-point metric
subregularity, an error-bound condition relating a measure of
first-order optimality to distance from the equilibrium set.
However, the same work gives a curved-domain example where OGDA
converges only polynomially. Thus, extending the classical
last-iterate guarantees requires understanding the geometry
of the strategy space.

To formulate such an extension precisely, one must also distinguish
rates over every game from rates for each fixed game.
Cai et al.~\cite{CaiEtAl2024} constructed families of classical
two-action games on which OMWU and a broader class of optimistic
follow-the-regularized-leader methods have a constant duality gap
at arbitrarily late iterations. These lower bounds remain consistent
with geometric convergence whose constants depend on the fixed game.
Su et al.~\cite{SuEtAl2026} transferred these examples to diagonal
quantum games. In the same work, they pursued geometric convergence
through semidefinite error bounds and conjectured geometric
last-iterate convergence of unregularized OMMWU whenever the Nash
equilibrium is unique. The diagonal lower bounds do not settle this conjecture.

Beyond diagonal embeddings, matrix strategy spaces introduce
additional geometric features. Ickstadt et al.~
\cite[Examples~3.5--3.6]{ITTV25} give $2\times 2$ semidefinite
games illustrating the roles of off-diagonal entries, rank-one
strategies, and strict complementarity. In particular, their
Example~3.6 exhibits a continuum of strictly complementary pure
equilibria. Strict complementarity therefore does not by itself
ensure equilibrium isolation, and assumptions on equilibrium
uniqueness must be distinguished from assumptions on the
associated slack matrices. These examples clarify the equilibrium
geometry without determining how quickly learning dynamics
approach equilibrium.

A further issue is that convergence of matrix strategies involves
both eigenvalues and eigenvectors. For continuous-time matrix
exponential dynamics, Lotidis et al.~\cite{LMB23,LMB23cdc}
decompose the induced state dynamics into a commutative component
governing the eigenvalues and a non-commutative component governing
the eigenvectors. The former evolves like mixed strategies under
classical FTRL, whereas the latter has no classical counterpart.
This decomposition identifies an additional aspect of convergence
that is absent from diagonal games: concentration of the eigenvalues
does not by itself control alignment with the equilibrium support.
Since these results concern non-optimistic continuous-time dynamics,
they leave open how optimism affects this eigenvector motion in
discrete time.

The literature therefore points to two potential limitations.
The uniform-average output may retain an accuracy barrier even
when the individual iterates converge rapidly. Separately, boundary
curvature and eigenvector motion may obstruct the extension of
classical last-iterate guarantees to matrix strategies. This leads
to the following questions:
\begin{quote}
\emph{Is the $\mathcal{O}(1/\varepsilon)$ dependence unavoidable
for the uniform-average output of OMMP that includes the initial
iterate? For each fixed quantum zero-sum game, do OGDA and OMMWU
inherit the geometric last-iterate convergence of their classical
counterparts, with uniqueness assumed for OMMWU?}
\end{quote}
We study these questions separately, distinguishing the output rule
from the underlying dynamics. For last-iterate convergence, we also
distinguish distance to equilibrium from the duality gap, which
measures the sum of the players' unilateral improvement opportunities.
Near a boundary equilibrium, these errors can decay at different
polynomial rates.

\subsection{Our contributions}

We give explicit, normalized games with one qubit per player. Every game
is fixed independently of the target accuracy. We study the two-projection
optimistic mirror-descent formulation of OGDA and the corresponding
entropy-based OMMWU updates. For OMMWU, entropy defines the mirror step,
while the target remains a Nash equilibrium of the original payoff.
All last-iterate statements concern the played strategy pair. The following
three results describe the limitations
of averaging and of these two choices of regularizer.

\paragraph{Tightness for uniform averaging.}
We construct a diagonal payoff observable with a unique pure equilibrium
such that every feasible sequence starting at the maximally mixed state
satisfies $G(\overline\Psi_T)\geq 1/T$, where $G$ is the duality gap and
$\overline\Psi_T= \tfrac{1}{T} \sum_{t=0}^{T-1}\Psi_t$. Consequently, this output needs
$\Omega(1/\varepsilon)$ iterations to be an $\varepsilon$-Nash equilibrium.
The bound applies independently of how the iterates are generated, and
hence independently of the OMMP regularizer and step size. It establishes
the conjectured accuracy dependence for the uniform-average return rule
in Algorithm~5 of Vasconcelos et al.~\cite{VasconcelosEtAl2025}, which
retains the initial state. The construction embeds a classical two-action
game; its obstruction is the averaging rule itself.

\paragraph{Polynomial last-iterate convergence of OGDA.}
We give a fixed game with a unique pure equilibrium $\Psi^*$ such that,
for every fixed step size $0<\eta\leq 1/8$, OGDA from
$\Psi_0=(I/2,I/2)$ satisfies
\begin{equation}
\operatorname{dist}_F(\Psi_t,\Psi^*)
\sim \frac{4}{\eta t},
\qquad
G(\Psi_t)\sim \frac{4}{\eta^3t^3}.
\label{eq:intro-ogda-rates}
\end{equation}
Thus the first iteration reaching distance $\varepsilon$ has order
$\varepsilon^{-1}$, whereas reaching gap $\varepsilon$ takes order
$\varepsilon^{-1/3}$. This extends the curved-domain construction of
Wei et al.~\cite[Theorem~9]{WeiEtAl2021} to full qubit strategy spaces and
establishes exact asymptotic rates from maximally mixed initialization.
A separate fixed game exhibits arbitrarily long delays across a family
of initial states, yielding worst-case iteration lower bounds over initializations of \(\Omega(1 / \sqrt{\varepsilon})\) for the gap and \(\Omega(1 / \varepsilon)\) for the distance.
Each fixed trajectory in that family still admits a
geometric bound, and the same game reaches equilibrium in finitely many
steps from the maximally mixed state. These additional results distinguish
a failure of uniformity from polynomial decay along one trajectory.

\paragraph{Polynomial OMMWU convergence despite strict complementarity.}
We construct a fixed game with a unique pure equilibrium $\Psi^*$
that is strictly complementary: each player's equilibrium slack is
positive definite on the orthogonal complement of the equilibrium
support. For every fixed $\eta>0$, unregularized OMMWU from
$\Psi_0=(I/2,I/2)$ satisfies
\begin{equation}
\operatorname{dist}_F(\Psi_t,\Psi^*)=\Theta(1/t),
\qquad
G(\Psi_t)=\Theta(1/t^{2}),
\qquad
S(\Psi^*\|\Psi_t)=\Theta(1/t),
\label{eq:intro-ommwu-rates}
\end{equation}
where $S$ denotes the sum of the two quantum relative entropies, directed
from the equilibrium states to the iterates. We determine the leading
constants and prove positive lower bounds at every iteration. The
corresponding first hitting times have orders $\varepsilon^{-1}$,
$\varepsilon^{-1/2}$, and $\varepsilon^{-1}$, respectively. Since the same game works for every fixed positive step size, choosing a smaller constant step does not restore geometric convergence.

These last-iterate results concern the unregularized game and the specified
updates. The $\Omega(1/\varepsilon)$ Nash-approximation bound comes from
uniform averaging; the fixed-initialization last-iterate examples have the
different gap complexities stated above. Thus the results identify
limitations of particular outputs and dynamics, rather than a universal
lower bound for all equilibrium algorithms. Numerical experiments in the
paper confirm the predicted powers and leading constants.

\subsection{Our techniques}

\paragraph{An affine gap preserves the initial error.}
The averaging construction has fixed, unique best responses for both
players. As a result, its duality gap is affine in the strategy pair.
The gap of a uniform average therefore equals the average of the gaps.
The initial gap is one, and all subsequent gaps are nonnegative, giving
$G(\overline\Psi_T)\geq 1/T$. Even immediate convergence of every later iterate
would leave exactly this contribution. This identity makes the proof
independent of the update rule.

\paragraph{Boundary curvature yields a quadratic recurrence.}
For OGDA, we use the Bloch representation, which identifies qubit density
matrices with the three-dimensional unit ball and turns Frobenius
projection into Euclidean projection of Bloch vectors. The payoff in our
fixed-initialization example creates off-diagonal entries at the first
step and keeps both players on the same two-dimensional trajectory.
After finitely many steps, the played and auxiliary states lie on the
boundary. Near equilibrium, the component of the update tangent to this
boundary is quadratic in the current error. Tracking both projections
gives a scalar recurrence
$u_{t+1}=u_t-(\eta/4)u_t^2+\mathcal{O}(u_t^3)$ for a positive auxiliary
coordinate $u_t$. Its reciprocal grows asymptotically by $\eta/4$ per
round, yielding the $1/t$ distance law. An exact best-response calculation
shows that the gap is cubic in this distance along the trajectory. The technique described above comes from quantizing the example in \cite[Theorem 9]{WeiEtAl2021}.

The initialization-dependent example isolates another consequence of
the same geometry. A pure state can have population $r$ outside the
equilibrium support and an off-diagonal entry of order $\sqrt r$.
The opponent initially supplies no correcting gradient and moves only
slowly. Projection then reduces the off-diagonal entry by a factor close
to one, producing a delay of at least order $1/(\eta\sqrt r)$. Controlling this
cumulative contraction proves these worst-case iteration lower bounds.

\paragraph{Eigenvalue concentration does not ensure eigenvector alignment.}
For OMMWU, we analyze the Hermitian \emph{score matrix} inside the
normalized matrix exponential. The identity
$\log(e^H/\operatorname{Tr}[e^H])=H-\log\operatorname{Tr}[e^H]I$
allows us to unroll the matrix updates exactly, including the optimistic
correction. Bob approaches his equilibrium state exponentially quickly.
His earlier deviations nevertheless leave a nonzero accumulated
off-diagonal term in Alice's score. Writing $X$ and $Z$ for the Pauli
matrices with off-diagonal and diagonal entries, respectively, Alice's
state has the form
\begin{equation}
\alpha_t=\frac{e^{H_t}}{\operatorname{Tr}[e^{H_t}]},
\qquad
H_t=\kappa tZ-p_tX,
\qquad
\kappa=\eta/3,
\qquad
p_t\longrightarrow p_\infty>0.
\label{eq:intro-score}
\end{equation}
The diagonal coefficient grows linearly, so the smaller eigenvalue of
$\alpha_t$ becomes exponentially small. However, exponentiation preserves
eigenvectors. The leading eigenspace of $H_t$ remains tilted away from
the equilibrium support by an angle of order $p_t/(\kappa t)$, and hence
aligns only at rate $1/t$. Exact formulas for the two-dimensional matrix
exponential then yield the three rates
in~\eqref{eq:intro-ommwu-rates}. The relevant conditional payoff operators
do not commute, so this evolution cannot be represented in a single fixed
diagonal basis. The example identifies eigenvector motion as an additional
obstruction that persists despite uniqueness and strict complementarity.

\subsection{Acknowledgment and AI disclosure}
The authors thank one of the anonymous reviewers of QTML 2026, whose comment motivates the lower bound constructions in this paper. The author P . X wants to thank Xiao Wang (Shanghai University of Finance and Economics) for related discussions.  The generative AI models (Fable 5.1 and ChatGPT6-Astra) were used to polish the writing and audit the proofs and arguments. Most of the ideas and mathematical content in this paper are due to the authors, except that the arguments for Theorem \ref{thm:ogda-mm-rates} are developed with the assistance of ChatGPT6-Astra. 

\section{Preliminaries and Notations}
\label{sec:preliminaries}

We study optimistic gradient descent--ascent (OGDA) and optimistic matrix
multiplicative weights updates (OMMWU) on several fixed games with one qubit
per player. These are optimistic matrix mirror-prox (OMMP) methods with the squared
Frobenius norm and negative von Neumann entropy as regularizers, respectively. Throughout, the term “last iterate” refers to the played strategy pair \(\joint_t\), as defined below.

% The OGDA construction
% gives lower bounds uniform over a family of initial states; the OMMWU
% construction gives polynomial convergence along a single trajectory
% from the maximally mixed state. Throughout, a last iterate means the played pair $\joint_t$.

\subsection{Density matrices and Bloch vectors}

Let $\DD_2=\{\rho\in\mathbb C^{2\times2}:\rho=\rho^\dagger,
\ \rho\succeq0,\ \Tr[\rho]=1\}$ and $\ZZ=\DD_2\times\DD_2$.
A density matrix is \emph{pure} if it has rank one, equivalently
$\Tr[\rho^2]=1$; the maximally mixed state is $I/2$.
Here $I$ is the $2\times2$ identity, and $\dagger$ denotes conjugate
transpose. On Hermitian matrices we use
$\langle H,K\rangle=\Tr[HK]$, the Frobenius norm
$\norm{H}_F=\sqrt{\Tr[H^2]}$, and the operator norm
$\norm{H}_{\rm op}=\max_j|\lambda_j(H)|$.
For pairs, $\norm{(H,K)}_F^2=\norm{H}_F^2+\norm{K}_F^2$.

Fix the Pauli matrices and basis projectors
\begin{equation}
X=\begin{pmatrix}0&1\\1&0\end{pmatrix},\quad
Y=\begin{pmatrix}0&-i\\i&0\end{pmatrix},\quad
Z=\begin{pmatrix}1&0\\0&-1\end{pmatrix},\quad
P=\frac{I+Z}{2},\quad Q=\frac{I-Z}{2}.
\label{eq:pauli}
\end{equation}
Thus $P=\operatorname{diag}(1,0)$ and $Q=\operatorname{diag}(0,1)$. Every qubit state has a unique Bloch representation
\begin{equation}
\rho(u,v,w)=\frac12(I+uX+vY+wZ)
=\frac12\begin{pmatrix}1+w&u-iv\\u+iv&1-w\end{pmatrix},
\qquad u^2+v^2+w^2\le1.
\label{eq:bloch}
\end{equation}
The vector $(u,v,w)$ is its Bloch vector; the state is pure exactly
when this vector has norm one. We write $D(b)=\operatorname{diag}(1-b,b)$
for $b\in[0,1]$. The population of a state in $Q$ is $\Tr[Q\rho]$,
and its off-diagonal entry is $\rho_{01}$, with indices $0,1$.
A nonzero off-diagonal entry will also be called coherence in this
fixed basis. The Pauli matrices satisfy $X^2=Z^2=I$ and $XZ+ZX=0$.
For matrices $H,K$, their commutator is $[H,K]=HK-KH$.

\subsection{Games, equilibria, and error measures}

A Hermitian observable $U$ with $\norm{U}_{\rm op}\le1$ defines
$f(\alpha,\beta)=\Tr[U(\alpha\otimes\beta)]$.
Alice minimizes $f$ over $\DD_2$, and Bob maximizes it over $\DD_2$.
The pair $\joint=(\alpha,\beta)\in\ZZ$ is distinct from the
tensor-product state $\alpha\otimes\beta$ used in the payoff.
Let $A(\beta)=\nabla_\alpha f(\alpha,\beta)$ and
$B(\alpha)=\nabla_\beta f(\alpha,\beta)$ be the gradients of this
bilinear function on the ambient real spaces of Hermitian matrices.
The saddle-point operator is $\gradient(\joint)=(A(\beta),-B(\alpha))$.
It is $L$-Lipschitz if
$\norm{\gradient(\joint)-\gradient(\joint')} _F
\le L\norm{\joint-\joint'}_F$ for all $\joint,\joint'\in\ZZ$.

A Nash equilibrium $\jointstar=(\alpha^*,\beta^*)$ satisfies
$f(\alpha^*,\beta)\le f(\alpha^*,\beta^*)\le f(\alpha,\beta^*)$
for every $\alpha,\beta\in\DD_2$. We denote the equilibrium set by
$\ZZ^*$ and its common payoff by $v_*$. At an equilibrium the slack
matrices are $W_A=A(\beta^*)-v_*I\succeq0$ and
$W_B=v_*I-B(\alpha^*)\succeq0$. The equilibrium is
\emph{strictly complementary} if
$\rank(\alpha^*)+\rank(W_A)=2$ and
$\rank(\beta^*)+\rank(W_B)=2$.

The duality gap and distance to equilibrium are
\begin{equation}
G(\alpha,\beta)=\max_{\beta'\in\DD_2}f(\alpha,\beta')
-\min_{\alpha'\in\DD_2}f(\alpha',\beta),
\qquad
\dist_F(\joint,\ZZ^*)=\inf_{\joint'\in\ZZ^*}\norm{\joint-\joint'}_F.
\label{eq:gap-definition}
\end{equation}
We use the convention that an $\varepsilon$-Nash equilibrium has
$G(\joint)\le\varepsilon$. This implies that neither player can
improve by more than $\varepsilon$; conversely, two individual
improvement bounds of $\varepsilon$ imply $G(\joint)\le2\varepsilon$.
When the equilibrium is unique, we abbreviate
$\dist_F(\joint,\{\jointstar\})$ as $\dist_F(\joint,\jointstar)$.

For $\rho\in\DD_2$ and positive definite $\sigma\in\DD_2$, the
quantum relative entropy is
$S(\rho\|\sigma)=\Tr[\rho\log\rho]-\Tr[\rho\log\sigma]$,
where $0\log0=0$ and all logarithms are natural.
For pairs with \(\alpha,\beta\succ0\), $S(\jointstar\|\joint)=S(\alpha^*\|\alpha)+S(\beta^*\|\beta)$.
The von Neumann entropy of $\rho$ is $-\Tr[\rho\log\rho]$.

For a run with initial state $\joint_0$, define the first hitting times
\begin{equation}
\begin{aligned}
\tau_{\rm gap}(\varepsilon;\joint_0)
&=\inf\{t\in\Nzero:G(\joint_t)\le\varepsilon\},\\
\tau_{\rm dist}(\varepsilon;\joint_0)
&=\inf\{t\in\Nzero:\dist_F(\joint_t,\ZZ^*)\le\varepsilon\},\\
\tau_S(\varepsilon;\joint_0)
&=\inf\{t\in\Nzero:S(\jointstar\|\joint_t)\le\varepsilon\},
\end{aligned}
\label{eq:hitting-times}
\end{equation}
with $\inf\varnothing=+\infty$; the last definition is used only for
the unique-equilibrium OMMWU construction. We omit $\joint_0$ when
it is fixed. 
% Asymptotic notation $\O$, $\Omega$, $\Theta$, $o$ has its usual meaning. 
We use standard asymptotic notation. For positive sequences \(a_t\) and \(b_t\), we write \(a_t\sim b_t\) if \(a_t/b_t\to1\) as \(t\to\infty\).
% We also use the asymptotic relation $\sim$ to describe the convergence rate. Specifically, we say $f(t) \sim g(t)$ if and only if $\lim_{t \rightarrow \infty} f(t)/g(t) = 1$. 
% The constants may depend on a fixed step size, but not on $t$ or the target accuracy $\varepsilon$.
Unless stated otherwise, implicit constants may depend on the fixed game, step size, and initialization, but not on the iteration index or target accuracy. Any uniformity over initial states is stated explicitly.
A geometric error estimate has the form $E_t\le C\lambda^t$ with
$C>0$ and $0<\lambda<1$.

\subsection{Optimistic updates}

Fix a step size $\eta>0$. Write the played and auxiliary pairs as
$\joint_t=(\alpha_t,\beta_t)$ and
$\widehat\joint_t=(\widehat\alpha_t,\widehat\beta_t)$, respectively,
and always initialize $\widehat\joint_0=\joint_0$.
All runs below continue without a stopping test. 

We first state the optimistic matrix mirror-prox (OMMP) updates
\cite[Algorithm~5]{VasconcelosEtAl2025}. Let $\varphi$ be a continuous, strongly convex regularizer on $\ZZ$,
differentiable at each auxiliary iterate. Its Bregman divergence is
\begin{equation}
D_\varphi(\joint',\joint)
=\varphi(\joint')-\varphi(\joint)
-\langle\nabla\varphi(\joint),\joint'-\joint\rangle,
\label{eq:ommp-bregman}
\end{equation}
whenever $\nabla\varphi(\joint)$ exists. The inner product on pairs
is the sum of the two Hilbert--Schmidt inner products.
In our minimization/maximization convention, OMMP performs
\begin{equation}
\begin{aligned}
\joint_{t+1}
&=\operatorname*{arg\,min}_{\joint'\in\ZZ}
\left\{\eta\langle\gradient(\joint_t),\joint'\rangle
       +D_\varphi(\joint',\widehat\joint_t)\right\},\\
\widehat\joint_{t+1}
&=\operatorname*{arg\,min}_{\joint'\in\ZZ}
\left\{\eta\langle\gradient(\joint_{t+1}),\joint'\rangle
       +D_\varphi(\joint',\widehat\joint_t)\right\}.
\end{aligned}
\label{eq:ommp-updates}
\end{equation}
The first step uses the previously computed operator value
$\gradient(\joint_t)$ to form the next played pair. Evaluating
$\gradient(\joint_{t+1})$ then provides the direction for updating
the auxiliary pair. Both minimizations use the same reference
point $\widehat\joint_t$. The played sequence $(\joint_t)_{t\ge0}$
is the sequence used in our last-iterate and average-iterate results.

The regularizer determines the geometry of these updates.
For $\varphi(\joint)=\norm{\joint}_F^2/2$, the divergence is
$D_\varphi(\joint',\joint)=\norm{\joint'-\joint}_F^2/2$, and
\eqref{eq:ommp-updates} provides the projected OGDA updates below.
For the negative von Neumann entropy
$\varphi(\alpha,\beta)=\Tr[\alpha\log\alpha]
+\Tr[\beta\log\beta]$, the divergence is the sum of the players'
quantum relative entropies, and \eqref{eq:ommp-updates} gives
OMMWU. In the latter case, we initialize both states as positive
definite, which keeps the auxiliary iterates in the differentiability
domain of $\varphi$.

Let $\Pi_{\ZZ}$ be Frobenius projection onto $\ZZ$, which separates
into projections $\Pi_{\DD_2}$ on its two factors. We use the
Euclidean optimistic mirror-descent formulation of OGDA:
\begin{equation}
\joint_{t+1}=\Pi_{\ZZ}(\widehat\joint_t-\eta\gradient(\joint_t)),
\qquad
\widehat\joint_{t+1}
=\Pi_{\ZZ}(\widehat\joint_t-\eta\gradient(\joint_{t+1})).
\label{eq:ogda-updates}
\end{equation}
Both projections are part of the algorithm; all OGDA results below
refer to this formulation \cite{WeiEtAl2021,SuEtAl2026}.

For Hermitian $H$, let $\Lambda(H)=e^H/\Tr[e^H]$. We use
unregularized OMMWU with updates
\begin{align}
\alpha_{t+1}&=\Lambda(\log\widehat\alpha_t-\eta A(\beta_t)),
&\beta_{t+1}&=\Lambda(\log\widehat\beta_t+\eta B(\alpha_t)),
\label{eq:played-update}\\
\widehat\alpha_{t+1}
&=\Lambda(\log\widehat\alpha_t-\eta A(\beta_{t+1})),
&\widehat\beta_{t+1}
&=\Lambda(\log\widehat\beta_t+\eta B(\alpha_{t+1})).
\label{eq:aux-update}
\end{align}
These are the entropy-based optimistic updates in
\cite[Algorithm~1]{VasconcelosEtAl2025} and
\cite[Algorithm~3]{SuEtAl2026}, expressed in our
minimization/maximization convention and indexed from $t=0$.
The payoff itself has no added entropy term. Positive definite
initial states give positive definite iterates, so the logarithms exist.
We use
$\Lambda(H+cI)=\Lambda(H)$ and
$\log\Lambda(H)=H-\log\Tr[e^H]I$ for real $c$.
After evaluating $\gradient(\joint_0)$, either method
requires one new operator evaluation per round.
% We will use
% $\Lambda(H+cI)=\Lambda(H)$ and
% $\log\Lambda(H)=H-\log\Tr[e^H]I$ for real $c$.
% After the initial evaluation of $\gradient(\joint_0)$, either method
% requires one new operator evaluation per round.

We reuse $f,A,B,\gradient,G,\ZZ^*$ and the iterate notation across constructions; each refers to the game and run currently under consideration. Each payoff observable is introduced separately.

% \subsection{Background on convergence rates}

% For classical zero-sum matrix games with a unique equilibrium,
% optimistic multiplicative weights has geometric last-iterate
% convergence under a suitable step-size condition
% \cite{WeiLeeZhangLuo2021}. This motivates asking whether the same
% conclusion extends to matrix-valued strategies.

% The regularizer can affect the rate near a boundary solution.
% For the scalar objective $s^2/2$ on $[0,\infty)$ and a positive
% initial value, entropic mirror descent gives
% $s_{t+1}=s_t e^{-\eta s_t}$ and $s_t\sim1/(\eta t)$, whereas
% Euclidean gradient descent is geometric for $0<\eta<1$.
% A nonzero boundary gradient can restore geometric entropic convergence;
% see Examples~2 and~7 of \cite{AzizianEtAl2024}. Their Theorem~2 gives
% faster convergence along sharp directions, meaning directions with
% a nonzero first-order increase away from the solution, under
% additional polyhedral and separability assumptions. These observations
% motivate examining the local matrix geometry; neither lower bound
% below relies on transferring that theorem to density matrices.

\section{An Average-Iterate Lower Bound for OMMP}
\label{sec:ommp-average-lower-bound}

% We construct a fixed game with one qubit per player for which the
% uniform-average output of optimistic matrix mirror-prox (OMMP) requires
% $\Omega(1/\varepsilon)$ iterations to be an $\varepsilon$-Nash equilibrium.
% The argument applies to every choice of regularizer and step size that
% produces feasible iterates. The average includes the maximally mixed
% initial iterate, as in the return statement of
We construct a fixed game with one qubit per player whose duality gap is affine in the strategy pair. For any feasible sequence starting from the maximally mixed state, the initial gap contributes \(1/T\) to the gap of the uniform average. This yields an \(\Omega(1/\varepsilon)\) iteration lower bound for the uniform-average output of OMMP, independently of the regularizer and step size. The average includes the maximally mixed
initial iterate, as in the return statement of
\cite[Algorithm~5]{VasconcelosEtAl2025}. 

\subsection{Game construction}
\label{subsec:ommp-average-game}

Using the projectors $P$ and $Q$ from~\eqref{eq:pauli}, define the
payoff observable
\begin{equation}
U_{\rm avg}
=-P\otimes I+I\otimes P-\frac12P\otimes P
=\operatorname{diag}\left(-\frac12,-1,1,0\right),
\label{eq:ommp-average-observable}
\end{equation}
where the diagonal is written in the computational basis
$(|00\rangle,|01\rangle,|10\rangle,|11\rangle)$.
The payoff is $f(\alpha,\beta)=\Tr[U_{\rm avg}(\alpha\otimes\beta)]$.
% This defines a finite-valued quantum zero-sum game: the referee measures
% with the four orthogonal projectors 
The referee implements this game by measuring with the four orthogonal projectors
$P\otimes P$, $P\otimes Q$,
$Q\otimes P$, and $Q\otimes Q$, whose sum is $I\otimes I$.
% The corresponding diagonal entry of $U_{\rm avg}$ is Alice's loss and
% Bob's payoff. 
The corresponding diagonal entry of \(U_{\rm avg}\) specifies Alice's loss and Bob's payoff.
All payoff values lie in $[-1,1]$, and
$\norm{U_{\rm avg}}_{\rm op}=1$.

For $\joint=(\alpha,\beta)\in\ZZ$, the populations
$\Tr[Q\alpha]$ and $\Tr[Q\beta]$ lie in $[0,1]$.
Using $\Tr[P\alpha]=1-\Tr[Q\alpha]$ and
$\Tr[P\beta]=1-\Tr[Q\beta]$, we obtain
\begin{equation}
f(\alpha,\beta)
=-\frac12+\frac32\Tr[Q\alpha]-\frac12\Tr[Q\beta]
-\frac12\Tr[Q\alpha]\Tr[Q\beta].
\label{eq:ommp-average-payoff}
\end{equation}
Each population can take any value in $[0,1]$, realized by a diagonal state.
Since the payoff is independent of off-diagonal entries, optimizing
over density matrices is equivalent to optimizing over these populations.

\begin{lemma}
\label{lem:ommp-average-gap}
The game in~\eqref{eq:ommp-average-observable} has the unique Nash
equilibrium $\jointstar=(P,P)$ and value $v_*=-1/2$. For every strategy
pair $\joint=(\alpha,\beta)\in\ZZ$, its duality gap is
\begin{equation}
G(\joint)=\frac32\Tr[Q\alpha]+\frac12\Tr[Q\beta].
\label{eq:ommp-average-affine-gap}
\end{equation}
In particular, $G$ is affine in the strategy pair $\joint$.
\end{lemma}

\begin{proof}
For fixed $\beta$, the coefficient of $\Tr[Q\alpha]$ in
\eqref{eq:ommp-average-payoff} is $(3-\Tr[Q\beta])/2\ge1$.
The payoff is therefore strictly increasing in $\Tr[Q\alpha]$,
so Alice's best responses satisfy
$\Tr[Q\alpha]=0$. For fixed $\alpha$, the payoff is strictly
decreasing in $\Tr[Q\beta]$, with coefficient
$-(1+\Tr[Q\alpha])/2\le-1/2$, so Bob's best responses satisfy
$\Tr[Q\beta]=0$.
A density matrix with zero population in $Q$ has diagonal entries
$1$ and $0$; positive semidefiniteness forces its off-diagonal entries
to vanish. It must therefore equal $P$.
Thus the unique equilibrium is $\jointstar=(P,P)$, and substitution
gives $v_*=-1/2$.

The best-response payoffs are
\begin{align*}
\min_{\alpha'\in\DD_2}f(\alpha',\beta)
&=f(P,\beta)=-\frac12-\frac12\Tr[Q\beta],\\
\max_{\beta'\in\DD_2}f(\alpha,\beta')
&=f(\alpha,P)=-\frac12+\frac32\Tr[Q\alpha].
\end{align*}
Subtracting these best-response payoffs gives
$G(\joint)=\frac32\Tr[Q\alpha]+\frac12\Tr[Q\beta]$,
which is affine in $\joint$.
\end{proof}

\subsection{Lower bound for the uniform average}
\label{subsec:ommp-average-rate}

For $T\ge1$, consider the output
$\overline\joint_T=(\bar\alpha_T,\bar\beta_T)$, where
\begin{equation}
\bar\alpha_T=\frac1T\sum_{t=0}^{T-1}\alpha_t,
\qquad
\bar\beta_T=\frac1T\sum_{t=0}^{T-1}\beta_t,
\qquad
\joint_0=(I/2,I/2).
\label{eq:ommp-average-output}
\end{equation}
Thus $\overline\joint_T=\tfrac{1}{T} \sum_{t=0}^{T-1}\joint_t\in\ZZ$.
Here $T$ counts the averaged iterates, so producing them requires
$T-1$ updates after initialization.

\begin{theorem}[Average-iterate lower bound for OMMP]
\label{thm:ommp-average-lower-bound}
For the game in~\eqref{eq:ommp-average-observable}, every feasible
sequence $(\joint_t)_{t\ge0}$ with the initialization
in~\eqref{eq:ommp-average-output} satisfies
\begin{equation}
G(\overline\joint_T)\ge\frac1T
\qquad\text{for every }T\ge1.
\label{eq:ommp-average-gap-lower-bound}
\end{equation}
Consequently, if $\overline\joint_T$ is an $\varepsilon$-Nash equilibrium,
then $T\ge1/\varepsilon$.
This yields an $\Omega(1/\varepsilon)$ update lower bound for the
uniform-average output of OMMP, independently of its regularizer
and step size.
\end{theorem}

\begin{proof}
At initialization, $\Tr[Q\alpha_0]=\Tr[Q\beta_0]=1/2$, so
\eqref{eq:ommp-average-affine-gap} gives $G(\joint_0)=1$.
By the affine dependence of the gap on the strategy pair and its
nonnegativity at every feasible pair,
\begin{equation}
G(\overline\joint_T)
=\frac1T\sum_{t=0}^{T-1}G(\joint_t)
\ge\frac1T G(\joint_0)
=\frac1T.
\label{eq:ommp-average-gap-calculation}
\end{equation}
Under the convention in Section~\ref{sec:preliminaries}, an
$\varepsilon$-Nash equilibrium satisfies $G(\overline\joint_T)\le\varepsilon$.
Hence $T\ge1/\varepsilon$, and the number of updates is at least
$\lceil1/\varepsilon\rceil-1=\Omega(1/\varepsilon)$ as
$\varepsilon\to0$.
\end{proof}

\begin{remark}[Scope of the lower bound]
The game, its dimension, and its initialization are fixed independently
of $T$ and $\varepsilon$. The obstruction comes from retaining the initial
state with weight $1/T$ in the output: even if every subsequent iterate
were exactly $\jointstar$, the averaged output would still have gap $1/T$.
Together with the known OMMP upper bound under its prescribed regularizer
and step-size conditions~\cite{VasconcelosEtAl2025}, this establishes the
optimal $1/\varepsilon$ dependence for the specified averaging rule.
The result concerns this output rule. It does not give a last-iterate
lower bound or a lower bound for a modified averaging rule. The observable
is diagonal and embeds a classical two-action zero-sum game, so the
obstruction does not rely on noncommutativity.
\end{remark}

\section{Last-Iterate Lower Bounds for OGDA}
\label{sec:ogda}
We give two fixed games that exhibit different limitations of OGDA. In the first game, the time needed to reduce the initial error by a constant factor becomes arbitrarily large as the initialization approaches equilibrium, although each fixed trajectory in this family admits a geometric convergence bound. In the second game, OGDA converges polynomially from the maximally mixed state, with equilibrium distance of order \(1/t\) and duality gap of order \(1/t^{3}\).
% In Section \ref{subsec:ogda-construction}--\ref{subsec:ogda-lowerbound-2}, we construct a family of (fixed) games and prove that the convergence delay of OGDA on such games becomes arbitrarily long as the initialization approaches a particular equilibrium. In Section \ref{subsec:ogda-mm}, we construct a fixed game that delays the convergence of OGDA when initializing at the maximally mixed state.
\subsection{Game construction}
\label{subsec:ogda-construction}

Consider the observable and payoff
\begin{equation}
U_{\rm G}=Q\otimes Q,\qquad
f(\alpha,\beta)=\Tr[Q\alpha]\Tr[Q\beta].
\label{eq:ogda-game}
\end{equation}
Its saddle-point operator is
$\gradient(\alpha,\beta)=(\Tr[Q\beta]Q,-\Tr[Q\alpha]Q)$.
For $0<r\le1/4$, the pure initial state and time horizon used below are
\begin{equation}
\rho_r=\begin{pmatrix}1-r&\sqrt{r(1-r)}\\
\sqrt{r(1-r)}&r\end{pmatrix},\qquad
T_r=\left\lfloor\frac{1}{2\eta\sqrt r}\right\rfloor.
\label{eq:ogda-initial-family}
\end{equation}
The initial gap is $r$. We will show that it remains above $r/2$ for
all $t\le T_r$, even though $T_r\to\infty$ as $r$ goes to 0.
The main OGDA run starts at $(\rho_r,P)$. We also examine the separate
run starting at $(I/2,I/2)$. In this construction $0<\eta\le1/8$. 

Similar semidefinite constructions appear in
\cite[Examples~3.5(b) and~3.6]{ITTV25}, which illustrate the roles
of rank-one strategies and off-diagonal entries in equilibrium
geometry. Moreover, $\rho_r$ is the rank-one matrix displayed in their
Example~3.6 with parameter $t=1-r$. In our construction, we use these geometric features to study the convergence of OGDA in a zero-sum game. In particular, the
initializations $(\rho_r,P)$ combine a small population outside
the equilibrium support with a substantially larger coherence.
Together with Bob's slow initial motion, this population--coherence
relation produces arbitrarily long convergence delays across
the family of initializations.

\begin{lemma}[Geometry of the OGDA game]
\label{lem:ogda-game}
The observable $U_{\rm G}$ satisfies $0\preceq U_{\rm G}\preceq I$
and $\norm{U_{\rm G}}_{\rm op}=1$. Its saddle-point operator is
$1$-Lipschitz. The game has value zero, and
\begin{equation}
\ZZ^*=\{P\}\times\DD_2,\qquad
G(\alpha,\beta)=\Tr[Q\alpha],\qquad
\dist_F((\alpha,\beta),\ZZ^*)=\norm{\alpha-P}_F.
\label{eq:ogda-gap-nash}
\end{equation}
If $\alpha$ is pure, then
\begin{equation}
\dist_F((\alpha,\beta),\ZZ^*)^2=2\Tr[Q\alpha].
\label{eq:ogda-pure-distance}
\end{equation}
\end{lemma}

\begin{proof}
The normalization follows because $Q\otimes Q$ is an orthogonal
projector. Since $\Tr[Q\alpha],\Tr[Q\beta]\in[0,1]$, the payoff is
nonnegative, and Alice guarantees zero by playing $P$. Her worst-case
payoff is $\Tr[Q\alpha]$, attained at $\beta=Q$, so she is optimal
precisely when $\Tr[Q\alpha]=0$. Positive
semidefiniteness then forces both off-diagonal entries of $\alpha$ to
vanish, giving $\alpha=P$. Every Bob strategy is optimal, since
Alice can attain zero against it. These observations give the value,
equilibrium set, and gap in \eqref{eq:ogda-gap-nash}. The distance
formula follows by keeping Bob's state unchanged in the closest
equilibrium pair. If $\alpha$ is pure, it yields
$$\norm{\alpha-P}_F^2=\Tr[\alpha^2]+\Tr[P^2]-2\Tr[\alpha P]
=2-2\Tr[P\alpha]=2\Tr[Q\alpha].$$

Finally, the linear extension of $\gradient$ obeys
$$\norm{\gradient(H,K)}_F^2
=|\Tr[QK]|^2+|\Tr[QH]|^2\le\norm{K}_F^2+\norm{H}_F^2.$$
Applying this inequality to differences proves the Lipschitz claim.
\end{proof}

\subsection{Why the construction delays convergence}
\label{subsec:ogda-lowerbound-1}
For the iterates, define the populations
\[
a_t=\Tr[Q\alpha_t],\qquad b_t=\Tr[Q\beta_t],\qquad
\widehat a_t=\Tr[Q\widehat\alpha_t],\qquad
\widehat b_t=\Tr[Q\widehat\beta_t],
\]
and the off-diagonal entries $c_t=(\alpha_t)_{01}$ and
$\widehat c_t=(\widehat\alpha_t)_{01}$.
Then \eqref{eq:ogda-updates} specializes to
\begin{align}
\alpha_{t+1}&=\Pi_{\DD_2}(\widehat\alpha_t-\eta b_tQ),
&\beta_{t+1}&=\Pi_{\DD_2}(\widehat\beta_t+\eta a_tQ),
\label{eq:ogda-played-components}\\
\widehat\alpha_{t+1}
&=\Pi_{\DD_2}(\widehat\alpha_t-\eta b_{t+1}Q),
&\widehat\beta_{t+1}
&=\Pi_{\DD_2}(\widehat\beta_t+\eta a_{t+1}Q).
\label{eq:ogda-aux-components}
\end{align}

The delay comes from two features that reinforce each other: Alice's
gradient is initially zero, and the curved boundary of the state space
makes her coherence decrease multiplicatively under projection.
The projection lemma and trajectory theorem below give the precise
bounds; we first explain their meaning.

\paragraph{A small population can coexist with a much larger coherence.}
For Alice's played state, write
\[
\alpha_t=
\begin{pmatrix}
1-a_t&c_t\\
c_t&a_t
\end{pmatrix}.
\]
Thus $a_t$ is the population in the $Q$ state, while $c_t$ is the
off-diagonal entry, or coherence, in the fixed $P,Q$ basis. Along this
trajectory $c_t$ is real and positive. At the initialization
$(\rho_r,P)$, we have $a_0=r$, $c_0=\sqrt{r(1-r)}$, and Bob's
population is zero. Since Alice's gradient is $b_0Q=0$, her first
played update leaves her at $\alpha_1=\rho_r$, although her gap is
$r>0$. Alice remains pure throughout the trajectory, so
$\det(\alpha_t)=0$ gives $c_t^2=a_t(1-a_t)$. In particular, since
$a_t\le r\le1/4$,
\[
\frac34a_t\le c_t^2\le a_t.
\]
The population and squared coherence are therefore comparable: a
constant-factor reduction in $a_t$ entails a constant-factor reduction
in $c_t^2$.

\paragraph{Bob's slow motion gives Alice only a weak contraction.}
Bob's played and auxiliary states remain diagonal. Write
$\beta_t=D(b_t)$ and $\widehat\beta_t=D(\widehat b_t)$.
For a diagonal state, Frobenius projection gives
\[
\Pi_{\DD_2}\bigl(D(b)+hQ\bigr)
=D\left(\min\left\{1,b+\frac h2\right\}\right)
\qquad(h\ge0).
\]
The factor $1/2$ comes from restoring trace one: before clipping at
the boundary, the projection subtracts $(h/2)I$ from $D(b)+hQ$.
Consequently, Bob's updates satisfy
\[
b_{t+1}
=\min\left\{1,\widehat b_t+\frac{\eta a_t}{2}\right\},
\qquad
\widehat b_{t+1}
=\min\left\{1,\widehat b_t+\frac{\eta a_{t+1}}{2}\right\}.
\]
Since every Alice population is at most $r$ and
$b_0=\widehat b_0=0$, each auxiliary update increases
$\widehat b_t$ by at most $\eta r/2$. Induction gives
$\widehat b_t\le\eta rt/2$ and, using the played recurrence,
$b_t\le\eta rt/2$.

To see how this affects Alice, consider one of her projections.
If the pure input has Bloch coordinates $(u,0,w)$ with $w\ge0$
and coherence $c=u/2$, then a shift of size $h\ge0$ gives
\begin{equation}
\Pi_{\DD_2}(\rho(u,0,w)-hQ)
=\rho\left(\frac{u}{N},0,\frac{w+h}{N}\right),
\qquad
N=\sqrt{1+2hw+h^2},\qquad c'=\frac{c}{N}.
\label{eq:ogda-mechanism-projection}
\end{equation}
Here $1\le N\le1+h$, so the coherence decreases by a factor no
smaller than $(1+h)^{-1}$. In the updates, $h=\eta b_t$ or
$h=\eta b_{t+1}$. Since $\log(1+h)\le h$, the cumulative decrease
is controlled by the sum of these shifts. Counting the auxiliary
updates and the additional played update gives
\begin{equation}
\log\frac{c_0}{c_t}
\le\eta\sum_{j=1}^{t-1}b_j+\eta b_{t-1}
\le\frac{\eta^2r}{4}(t-1)(t+2)
\qquad(t\ge1).
\label{eq:ogda-mechanism-budget}
\end{equation}
Now suppose that $c_t\le q c_0$ for some fixed $q\in(0,1)$.
By \eqref{eq:ogda-mechanism-budget} and
$(t-1)(t+2)\le2t^2$,
\[
\log\frac1q
\le\log\frac{c_0}{c_t}
\le\frac{\eta^2rt^2}{2}.
\]
Hence $t\ge\sqrt{2\log(1/q)}/(\eta\sqrt r)$: a fixed reduction
of coherence requires $\Omega(1/(\eta\sqrt r))$ iterations. The
relations $G(\joint_t)=a_t\ge c_t^2$ and
$\dist_F(\joint_t,\ZZ^*)^2=2a_t$ turn this delay into the stated
gap and distance lower bounds. Choosing $r$ proportional to
$\varepsilon$ gives an $\Omega(\varepsilon^{-1/2})$ gap bound;
choosing $r$ proportional to $\varepsilon^2$ gives an
$\Omega(\varepsilon^{-1})$ distance bound, for fixed $\eta$.

\paragraph{A simple illustration.}
The first round already exhibits the effect:
$\alpha_1=\rho_r$ and $b_1=\widehat b_1=\eta r/2$.
The first auxiliary Alice update consequently uses only the shift
$h=\eta^2r/2$, and its coherence satisfies
\begin{equation}
\frac{\widehat c_1}{c_0}
=\left(1+\eta^2r(1-2r)+\frac{\eta^4r^2}{4}\right)^{-1/2}
=1-\frac{\eta^2r}{2}+\O(r^2)
\quad\text{as }r\downarrow0
\label{eq:ogda-first-round-illustration}
\end{equation}
with $\eta$ fixed. Hence this ratio tends to one as the initial state
approaches $P$. Figure~\ref{fig:ogda-projection} shows the geometry
of one such projection.

\begin{figure}[htbp]
\centering
\begin{tikzpicture}[x=3.5cm,y=3.5cm,>=Stealth,font=\small]
  \coordinate (O) at (0,0);
  \coordinate (V) at (0.6,0.8);
  \coordinate (W) at (0.6,1.08);
  \coordinate (Vp) at ({0.6/sqrt(0.6*0.6+1.08*1.08)},
                        {1.08/sqrt(0.6*0.6+1.08*1.08)});
  % Project the two points vertically onto the u-axis.
  \coordinate (U) at (V |- O);
  \coordinate (Up) at (Vp |- O);

  \draw[densely dashed,constructionorange!65] (V) -- (U);
  \draw[densely dashed,constructionblue!65] (Vp) -- (Up);

  % Axis ticks.
  \draw[constructionorange,thick]
    ([yshift=2pt]U) -- ([yshift=-2pt]U);
  \draw[constructionblue,thick]
    ([yshift=2pt]Up) -- ([yshift=-2pt]Up);

  % Place the labels on opposite sides to avoid overlap.
  \node[below right=3pt,text=constructionorange] at (U)
    {$u=2c$};
  \node[below left=3pt,text=constructionblue] at (Up)
    {$u'=2c'$};

  \draw[->,gray!70] (-0.04,0) -- (1.14,0)
    node[right,text=black] {$u$};
  \draw[->,gray!70] (0,-0.03) -- (0,1.20)
    node[above,text=black] {$w$};

  \draw[gray!60,thick]
    (1,0) arc[start angle=0,end angle=90,radius=1];
  \draw[densely dashed,gray!70] (O) -- (W);
  \draw[->,constructionorange,thick] (V) -- (W)
    node[midway,right,text=black]{};
  \draw[->,constructionblue,very thick] (W) -- (Vp);

  \fill (0,1) circle[radius=0.012]
    node[left] {$P$};
  \fill[constructionorange] (V) circle[radius=0.013]
    node[right=8pt,text=brown] {$(u,w)$};
  \fill[constructionblue] (Vp) circle[radius=0.013]
    node[below left,xshift=-3pt,yshift=-1pt,text=constructionblue]
      {$(u',w')$};
  \fill[gray!80] (W) circle[radius=0.013]
    node[above right,text=brown] {$(u,w+h)$};
\end{tikzpicture}
\caption{One Alice projection for OGDA in the $(u,w)$ cross-section
of the Bloch sphere. The pure input satisfies $u^2+w^2=1$.
Subtracting $hQ$, where $h\ge0$, shifts the Bloch vector from $(u,w)$
to $(u,w+h)$ along the orange arrow. The blue arrow is the radial
projection onto the unit circle, giving
$(u',w')=(u/N,(w+h)/N)$ with $N=\sqrt{1+2hw+h^2}$.
Writing $c=u/2$ and $c'=u'/2$ for the input and output coherences,
we have $c'=c/N\ge c/(1+h)$. Thus a small shift changes the
coherence only slightly.}
\label{fig:ogda-projection}
\end{figure}
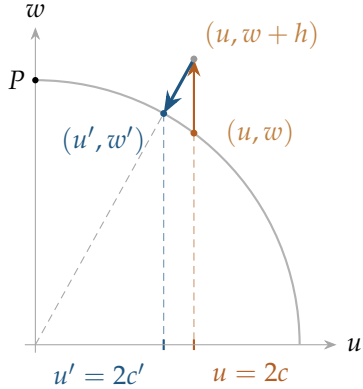
\FloatBarrier

\subsubsection{Projection on $\mathcal{D}_2$}
\label{subsec:ogda-projection}

The next lemma proves the formula used above, including the diagonal
case needed for Bob's updates.

\begin{lemma}[Qubit projection]
\label{lem:ogda-qubit-projection}
Let $\Pi_{\DD_2}$ denote Frobenius projection onto $\DD_2$.
For every qubit state $\rho(u,v,w)$ and $h\in\mathbb R$, we have
\begin{equation}
\Pi_{\DD_2}(\rho(u,v,w)-hQ)
=\rho\left(\frac{u}{L},\frac{v}{L},\frac{w+h}{L}\right),
\qquad
L=\max\left\{1,\sqrt{u^2+v^2+(w+h)^2}\right\}.
\label{eq:ogda-general-projection}
\end{equation}
If $h\ge0$, $v=0$, $u>0$, $w\ge0$, and $u^2+w^2=1$,
the projected state $\rho'$ remains real and pure, with
$\Tr[Q\rho']\le(1-w)/2$.
Let $c=u/2$. The off-diagonal entry $c'$ of $\rho'$ satisfies that
\begin{equation}
c'=\frac{c}{N},
\quad
1\le N:=\sqrt{1+2hw+h^2}\le1+h,
\quad \mathrm{and} \quad
\frac{c}{1+h}\le c'\le c.
\label{eq:ogda-offdiagonal-projection}
\end{equation}
For a diagonal state $D(b)$ with $b\in[0,1]$ and $h\ge0$,
the ascent projection is
\begin{equation}
\Pi_{\DD_2}(D(b)+hQ)
=D\left(\min\left\{1,b+\frac h2\right\}\right).
\label{eq:ogda-diagonal-projection}
\end{equation}
\end{lemma}

\begin{proof}
Every density matrix in $\DD_2$ has the form $\rho(u',v',w')$
with $u'^2+v'^2+w'^2\le1$. Since $Q=(I-Z)/2$, we have
\begin{equation*}
\rho(u',v',w')-(\rho(u,v,w)-hQ)
=\frac12\bigl(hI+(u'-u)X+(v'-v)Y+(w'-w-h)Z\bigr).
\end{equation*}
The matrices $I,X,Y,Z$ are mutually orthogonal in the Frobenius
inner product, and each has squared Frobenius norm $2$. Hence
\begin{equation*}
\begin{aligned}
\|\rho(u',v',w')-(\rho(u,v,w)-hQ)\|_F^2
=\frac12\left(h^2+(u'-u)^2+(v'-v)^2+(w'-w-h)^2\right).
\end{aligned}
\end{equation*}
The term $h^2/2$ does not depend on $(u',v',w')$. Thus, minimizing
this expression amounts to projecting $r=(u,v,w+h)$ onto the
Euclidean unit ball. If $\|r\|_2\le1$, the minimizing vector
is $r$. Otherwise, every feasible vector $s$ satisfies
$\|s-r\|_2\ge\|r\|_2-\|s\|_2\ge\|r\|_2-1$, with equality
at $s=r/\|r\|_2$. The projection is therefore
$r/\max\{1,\|r\|_2\}$, proving
\eqref{eq:ogda-general-projection}.

Under the assumptions for a real pure input, $v=0$ and
$u^2+w^2=1$, so the shifted vector has norm
$$N=\sqrt{u^2+(w+h)^2}=\sqrt{1+2hw+h^2}.$$
Since $0\le w\le1$ and $h\ge0$,
$1\le N^2\le(1+h)^2$, and hence $1\le N\le1+h$.
Thus $L=N$, and the projected Bloch vector
$(u/N,0,(w+h)/N)$ has unit norm and zero $Y$-coordinate.
The projected state $\rho'$ is therefore again real and pure.
Moreover,
\begin{equation*}
\frac{w+h}{N}
\ge\frac{w+h}{1+h}
=w+\frac{h(1-w)}{1+h}
\ge w.
\end{equation*}
Its population in $Q$ consequently satisfies
$\Tr[Q\rho']=\tfrac{1}{2}(1-(w+h)/N)/2\le(1-w)$.
The output off-diagonal entry is
$c'=u/(2N)=c/N$, where $c=u/2>0$.
Combining this identity with $1\le N\le1+h$ gives
$c/(1+h)\le c'\le c$, proving
\eqref{eq:ogda-offdiagonal-projection}.

For the last assertion, write $D(b)=\rho(0,0,1-2b)$ and
replace $h$ by $-h$ in \eqref{eq:ogda-general-projection}.
The relevant shifted Bloch vector is $(0,0,1-2b-h)$.
Since $b,h\ge0$, its third coordinate is at most $1$.
Projection onto the unit ball leaves this coordinate unchanged
when it is at least $-1$, and replaces it by $-1$ otherwise.
Thus the projected state remains diagonal, with third coordinate of
$$w'=\max\{-1,1-2b-h\}$$ and population
\begin{equation*}
b'=\frac{1-w'}{2}
=\frac{1-\max\{-1,1-2b-h\}}{2}
=\min\left\{1,b+\frac h2\right\}.
\end{equation*}
The projected state is $D(b')$, which proves
\eqref{eq:ogda-diagonal-projection}.
\end{proof}

\subsection{Proof of lower bounds}
\label{subsec:ogda-lowerbound-2}
In this section, we prove the lower bound on iteration complexity of OGDA on the game instance \eqref{eq:ogda-game}.
\subsubsection{Delay along trajectories}
\begin{theorem}[Lower bound along the OGDA trajectory]
\label{thm:ogda-trajectory-lower-bound}
Fix $0<\eta\le1/8$ and $0<r\le1/4$. Initialize
\begin{equation}
\alpha_0=\widehat\alpha_0=\rho_r,
\qquad
\beta_0=\widehat\beta_0=P.
\label{eq:ogda-hard-initialization}
\end{equation}
For the updates \eqref{eq:ogda-updates}, every $\alpha_t$ is pure,
and for all integers $t\ge1$,
\begin{equation}
G(\joint_t)
\ge r(1-r)
\exp\left(-\frac{\eta^2r}{2}(t-1)(t+2)\right),
\qquad
\dist_F(\joint_t,\ZZ^*)^2=2G(\joint_t).
\label{eq:ogda-trajectory-lower-bound}
\end{equation}
In particular, throughout the horizon $T_r$ in
\eqref{eq:ogda-initial-family},
\begin{equation}
G(\joint_t)>\frac r2,
\qquad
\dist_F(\joint_t,\ZZ^*)>\sqrt r,
\qquad
0\le t\le T_r.
\label{eq:ogda-delay-bound}
\end{equation}
\end{theorem}

\begin{proof}
We keep the played and auxiliary sequences separate throughout the
argument.

\paragraph{Step 1: Preserve the state structure.}
The initial Alice state is pure and has Bloch vector
$$(2\sqrt{r(1-r)},0,1-2r),$$ whose first coordinate is positive and
whose third coordinate is at least $1/2$. Both initial Bob states are diagonal.
Every iterate is a density matrix, so all populations are nonnegative.
Consequently, each Alice update in
\eqref{eq:ogda-played-components}--\eqref{eq:ogda-aux-components}
has the form in \eqref{eq:ogda-offdiagonal-projection} with $h\ge0$.
Inductively, $\alpha_t$ and $\widehat\alpha_t$ remain real pure states
with positive off-diagonal entries and third coordinates at least
$1-2r$. Each Alice projection weakly decreases population relative
to its auxiliary input. Therefore
\begin{equation}
0\le\widehat a_{t+1}\le\widehat a_t\le r,
\qquad
0\le a_{t+1}\le\widehat a_t\le r.
\label{eq:ogda-alice-population-bound}
\end{equation}
Likewise, \eqref{eq:ogda-diagonal-projection} shows inductively that
$\beta_t$ and $\widehat\beta_t$ remain diagonal. These conclusions
apply in the order of the updates: first construct the two played
states, then the two auxiliary states.

\paragraph{Step 2: Bound the growth of Bob's population.}
Equation \eqref{eq:ogda-diagonal-projection} gives the exact recurrences
\begin{equation}
b_{t+1}=\min\left\{1,\widehat b_t+\frac{\eta a_t}{2}\right\},
\qquad
\widehat b_{t+1}
=\min\left\{1,\widehat b_t+\frac{\eta a_{t+1}}{2}\right\}.
\label{eq:ogda-bob-recurrences}
\end{equation}
By \eqref{eq:ogda-alice-population-bound},
$\widehat b_{t+1}\le\widehat b_t+\eta r/2$.
Starting from $\widehat b_0=0$, induction yields
$\widehat b_t\le\eta rt/2$. The first recurrence then gives
$b_{t+1}\le\eta r(t+1)/2$. Including $b_0=0$, we have
\begin{equation}
0\le\widehat b_t\le\frac{\eta rt}{2},
\qquad
0\le b_t\le\frac{\eta rt}{2}
\qquad(t\ge0).
\label{eq:ogda-bob-growth-bound}
\end{equation}
These upper bounds remain valid when the clipping in
\eqref{eq:ogda-bob-recurrences} is active.

\paragraph{Step 3: Bound the decrease of Alice's off-diagonal entry.}
The positive entries $c_t$ and $\widehat c_t$ satisfy useful
multiplicative bounds. The two Alice updates and \eqref{eq:ogda-offdiagonal-projection}
imply
\begin{equation}
c_{t+1}\ge\frac{\widehat c_t}{1+\eta b_t},
\qquad
\widehat c_{t+1}
\ge\frac{\widehat c_t}{1+\eta b_{t+1}}.
\label{eq:ogda-coherence-recurrences}
\end{equation}
Iterating the second inequality gives
$\widehat c_{t-1}\ge
c_0\prod_{j=1}^{t-1}(1+\eta b_j)^{-1}$ for $t\ge1$.
Applying the first inequality for the played state at time $t$,
and using $1+u\le e^u$ for $u\ge0$, we obtain
\begin{align}
c_t
\ge\frac{c_0}
{(1+\eta b_{t-1})\prod_{j=1}^{t-1}(1+\eta b_j)}\ge c_0\exp\left(-\eta b_{t-1}
                      -\eta\sum_{j=1}^{t-1}b_j\right).
\label{eq:ogda-coherence-product}
\end{align}
When $t=1$, the product is empty and $b_0=0$, so this formula also
covers the first played update. The extra factor involving $b_{t-1}$
comes from this played update in addition to the preceding auxiliary
updates.

By \eqref{eq:ogda-bob-growth-bound}, the exponent satisfies
\begin{align*}
\eta b_{t-1}+\eta\sum_{j=1}^{t-1}b_j
\le\frac{\eta^2r}{2}
      \left((t-1)+\sum_{j=1}^{t-1}j\right)
=\frac{\eta^2r}{2}
      \left((t-1)+\frac{t(t-1)}2\right)
 =\frac{\eta^2r}{4}(t-1)(t+2).
\end{align*}
Since $c_0=\sqrt{r(1-r)}$, it follows that
\begin{equation}
c_t\ge\sqrt{r(1-r)}
\exp\left(-\frac{\eta^2r}{4}(t-1)(t+2)\right).
\label{eq:ogda-coherence-lower-bound}
\end{equation}
The determinant condition for the density matrix
$\alpha_t=\left(\begin{smallmatrix}1-a_t&c_t\\c_t&a_t\end{smallmatrix}\right)$
gives $c_t^2\le a_t(1-a_t)\le a_t$. Squaring
\eqref{eq:ogda-coherence-lower-bound} and using
$G(\joint_t)=a_t$ proves the gap lower bound in
\eqref{eq:ogda-trajectory-lower-bound}. Purity of $\alpha_t$ and
\eqref{eq:ogda-pure-distance} prove the distance identity.

\paragraph{Step 4: Obtain a constant-factor delay.}
For integers $t\ge1$, we have $(t-1)(t+2)\le2t^2$. If
$t\le1/(2\eta\sqrt r)$, then
\begin{equation*}
\frac{\eta^2r}{2}(t-1)(t+2)
\le\eta^2rt^2\le\frac14.
\end{equation*}
Thus, using $r\le1/4$ and $e^{-1/4}\ge1-1/4=3/4$,
\begin{equation*}
G(\joint_t)\ge r(1-r)e^{-1/4}
\ge\frac{9r}{16}>\frac r2.
\end{equation*}
The same strict inequality holds at $t=0$, since $G(\joint_0)=r$.
Finally, $\dist_F(\joint_t,\ZZ^*)^2=2G(\joint_t)>r$ gives the
distance assertion.
\end{proof}

\subsubsection{Iteration complexity and dependence on initialization}

The hitting times in \eqref{eq:hitting-times} turn the delay into
accuracy-dependent lower bounds.

\begin{corollary}[Iteration lower bounds over initial states]
\label{cor:ogda-iteration-lower-bounds}
For the fixed game \eqref{eq:ogda-game} and
$0<\eta\le1/8$, the following statements hold.
\begin{enumerate}
\item For every $0<\varepsilon\le1/16$, the initial state
$\joint_0=(\rho_{4\varepsilon},P)$ satisfies
\begin{equation}
\tau_{\rm gap}(\varepsilon;\joint_0)
\ge\left\lfloor\frac{1}{4\eta\sqrt\varepsilon}\right\rfloor+1.
\label{eq:ogda-gap-complexity}
\end{equation}
\item For every $0<\varepsilon\le1/4$, the initial state
$\joint_0=(\rho_{4\varepsilon^2},P)$ satisfies
\begin{equation}
\tau_{\rm dist}(\varepsilon;\joint_0)
\ge\left\lfloor\frac{1}{4\eta\varepsilon}\right\rfloor+1.
\label{eq:ogda-distance-complexity}
\end{equation}
\end{enumerate}
The game, dimension, payoff normalization, and step size are fixed;
only the initialization depends on $\varepsilon$.
\end{corollary}

\begin{proof}
For the gap claim, set $r=4\varepsilon\le1/4$. For every integer
$0\le t\le\lfloor1/(4\eta\sqrt\varepsilon)\rfloor$,
Theorem~\ref{thm:ogda-trajectory-lower-bound} gives
$G(\joint_t)>r/2=2\varepsilon$. Therefore the first iterate with
gap at most $\varepsilon$ must have an index at least one larger.

For the distance claim, set $r=4\varepsilon^2\le1/4$.
For every integer $0\le t\le\lfloor1/(4\eta\varepsilon)\rfloor$,
the same theorem gives
$\dist_F(\joint_t,\ZZ^*)>\sqrt r=2\varepsilon$.
This proves the second claim. In both cases the starting point itself
is outside the target tolerance, so including $t=0$ in the definition
of the first hitting time causes no ambiguity.
\end{proof}

The preceding quantifiers distinguish a uniform convergence guarantee
from the behavior of a single fixed trajectory. We make both points
explicit.

\begin{corollary}[No uniform geometric distance estimate]
\label{cor:ogda-no-uniform-geometric}
For the fixed game \eqref{eq:ogda-game} and any fixed
$0<\eta\le1/8$, there are no constants $M>0$ and $\lambda\in(0,1)$,
independent of the initial state, such that
\begin{equation}
\dist_F(\joint_t,\ZZ^*)^2
\le M\lambda^t\dist_F(\joint_0,\ZZ^*)^2
\qquad\text{for every initialization and every }t\ge0.
\label{eq:ogda-impossible-uniform-rate}
\end{equation}
\end{corollary}

\begin{proof}
For the initialization \eqref{eq:ogda-hard-initialization}, the initial
squared distance is $2r$. At
$T_r$, Theorem~\ref{thm:ogda-trajectory-lower-bound}
gives squared distance strictly greater than $r$. Thus
\eqref{eq:ogda-impossible-uniform-rate} would imply
$1/2<M\lambda^{T_r}$. But $T_r\to\infty$ as $r\downarrow0$, so
the right-hand side tends to zero, a contradiction.
\end{proof}

For each fixed $r$, however, the same trajectory does admit a geometric
upper bound. The next proposition makes clear that the lower bound
concerns uniformity over starting states.

\begin{proposition}[Geometric bound for a fixed starting state]
\label{prop:ogda-fixed-start-geometric}
Under the assumptions and initialization of
Theorem~\ref{thm:ogda-trajectory-lower-bound},
\begin{equation}
G(\joint_t)
\le r\left(1+\frac{\eta^2r}{2}\right)^{-(t-1)}
\qquad(t\ge1).
\label{eq:ogda-fixed-start-upper-bound}
\end{equation}
\end{proposition}

\begin{proof}
Because $b_0=0$, the first played Alice state is $\alpha_1=\alpha_0$,
so $a_1=r$. Equation \eqref{eq:ogda-bob-recurrences} then gives
$b_1=\widehat b_1=\eta r/2$; no clipping occurs because
$\eta r/2\le1/64$. The auxiliary Bob populations are nondecreasing,
and $b_{t+1}\ge\widehat b_t$. It follows that
$b_t\ge\eta r/2$ for every $t\ge1$.

In the auxiliary Alice update at round $t$, the shift has size
$\eta b_{t+1}\ge\eta^2r/2$, and the pure input has third coordinate
$1-2\widehat a_t\ge1/2$. By
\eqref{eq:ogda-offdiagonal-projection}, the squared normalization
factor is at least
\begin{equation*}
1+2\eta b_{t+1}(1-2\widehat a_t)+\eta^2b_{t+1}^2
\ge1+\eta b_{t+1}\ge1+\frac{\eta^2r}{2}.
\end{equation*}
Consequently,
$\widehat c_t^{\,2}\le r(1-r)(1+\eta^2r/2)^{-t}$.
The played projection only decreases the off-diagonal entry, so
$c_t\le\widehat c_{t-1}$ for $t\ge1$.
The state $\alpha_t$ is a real, rank-one density matrix, so
\begin{equation*}
0=\det
\begin{pmatrix}
1-a_t&c_t\\
c_t&a_t
\end{pmatrix}
=a_t(1-a_t)-c_t^2.
\end{equation*}
Thus $a_t=c_t^2/(1-a_t)$. Since $a_t\le r<1$,
the denominator satisfies $1-a_t\ge1-r>0$.
Combining this with $0<c_t\le\widehat c_{t-1}$ and the
preceding bound on $\widehat {c}_{t-1}^{\,2}$ gives
\begin{equation*}
a_t
=\frac{c_t^2}{1-a_t}
\le\frac{\widehat {c}_{t-1}^{\,2}}{1-r}
\le r\left(1+\frac{\eta^2r}{2}\right)^{-(t-1)}.
\end{equation*}
Since the duality gap equals $a_t$ by Lemma \ref{lem:ogda-game}, this proves the claim.
\end{proof}

For the maximally mixed initialization, the same game reaches an
exact equilibrium in finitely many iterations.

\begin{proposition}[Finite termination from maximally mixed states]
\label{prop:ogda-mixed-initialization}
For \eqref{eq:ogda-game}, initialize
$\alpha_0=\widehat\alpha_0=\beta_0=\widehat\beta_0=I/2$.
Then the updates \eqref{eq:ogda-updates} satisfy
\begin{equation}
G(\joint_t)=a_t
\le\max\left\{0,\frac12-\frac{\eta t}{4}\right\}
\qquad(t\ge0).
\label{eq:ogda-mixed-finite-bound}
\end{equation}
In particular, an exact equilibrium is reached by iteration
$\lceil2/\eta\rceil$.
\end{proposition}

\begin{proof}
All four state sequences remain diagonal by
\eqref{eq:ogda-general-projection}. Bob's populations obey
\eqref{eq:ogda-bob-recurrences} with initial value $1/2$.
Since all Alice populations are nonnegative, induction gives
$b_t,\widehat b_t\ge1/2$ for all $t$.

For a diagonal Alice state, the same projection formula gives
\begin{equation*}
a_{t+1}=\max\left\{0,\widehat a_t-\frac{\eta b_t}{2}\right\},
\qquad
\widehat a_{t+1}
=\max\left\{0,\widehat a_t-\frac{\eta b_{t+1}}{2}\right\}.
\end{equation*}
Both right-hand sides are at most
$\max\{0,\widehat a_t-\eta/4\}$.
Starting from $\widehat a_0=1/2$, induction shows
$\widehat a_t\le\max\{0,1/2-\eta t/4\}$ and the same bound for
$a_t$. Equation \eqref{eq:ogda-gap-nash} then proves
\eqref{eq:ogda-mixed-finite-bound}. When the right-hand side is zero,
$\alpha_t=P$, and every pair $(P,\beta_t)$ is an equilibrium.
\end{proof}

\begin{remark}[Scope of the result]
The lower bound applies to \eqref{eq:ogda-updates} with the
coherent initial states in \eqref{eq:ogda-hard-initialization}. It
does not disprove a logarithmic bound for an algorithm restricted to
the maximally mixed starting state, and it does not disprove a
geometric rate whose constants may depend on the starting state.
Propositions~\ref{prop:ogda-fixed-start-geometric}
and~\ref{prop:ogda-mixed-initialization} establish these distinctions
directly in the same game. 
\end{remark}

% Insert at the end of the OGDA section, before the OMMWU section.
% Uses the draft's existing notation, theorem environments, and bibliography.
% No preamble additions are required. All new labels use the ogda-mm prefix.

\subsection{A lower bound from maximally mixed initialization}
\label{subsec:ogda-mm}

Proposition~\ref{prop:ogda-mixed-initialization} shows that the game
\eqref{eq:ogda-game} reaches equilibrium in finitely many iterations
from $(I/2,I/2)$. We now construct a second game for which this same
initialization produces polynomial convergence along a single fixed
trajectory. The construction adapts the skew-symmetric example with
a curved boundary in
\cite[Theorem~9]{WeiEtAl2021} to the full
qubit state spaces. In this subsection, $f,A,B,\gradient,G,\ZZ^*$
and the iterate symbols refer to the new game and its run.

\paragraph{Game construction.}
Consider the observable
\begin{equation}
U_{\rm G}^{\rm mix}
=\frac14\big((I-Z)\otimes X-X\otimes(I-Z)\big).
\label{eq:ogda-mm-game}
\end{equation}
It is Hermitian and satisfies $\norm{U_{\rm G}^{\rm mix}}_{\rm op}\le1$,
since $\norm{I-Z}_{\rm op}=2$ and $\norm{X}_{\rm op}=1$.
For $x=\Tr[X\alpha]$, $z=\Tr[Z\alpha]$,
$y=\Tr[X\beta]$, and $w=\Tr[Z\beta]$, its payoff and gradients are
\begin{align}
f(\alpha,\beta)&=\frac14\big((1-z)y-x(1-w)\big),
\label{eq:ogda-mm-payoff}\\
A(\beta)&=\frac14\big(y(I-Z)-(1-w)X\big),\\
B(\alpha)&=\frac14\big((1-z)X-x(I-Z)\big).
\label{eq:ogda-mm-gradients}
\end{align}
Fix $0<\eta\le1/8$, use the updates~\eqref{eq:ogda-updates}, and
initialize $\widehat\joint_0=\joint_0=(I/2,I/2)$. Set $h=\eta/2$,
so $0<h\le1/16$. Neither the game nor the initialization depends
on the target accuracy.

\paragraph{Why the construction delays convergence.}
The payoff creates coherence at the first update:
$A(I/2)=-X/4$ and $B(I/2)=X/4$, so
$\alpha_1=\beta_1=\tfrac12(I+(\eta/2)X)$.
Both off-diagonal entries are therefore $\eta/4>0$.
The identity $B(\rho)=-A(\rho)$ keeps the two players' states
identical throughout the run. After a finite initial phase, all
played and auxiliary states become pure, as proved below.

To illustrate the subsequent motion, write a pure state on the
trajectory as $\rho(\sin\theta,0,\cos\theta)$, with
$0<\theta<\pi/2$. In the $(x,z)$ cross-section, its Bloch vector
is $r=(\sin\theta,\cos\theta)$ and the equilibrium is $e=(0,1)$.
The unprojected increment has direction
$g(r)=(1-\cos\theta,\sin\theta)$. Since the boundary of Bloch sphere is parameterized by $r(\theta)=(\sin\theta,\cos\theta)$,
differentiating with respect to $\theta$ gives
$r'(\theta)=(\cos\theta,-\sin\theta)$.
Since $\|r'(\theta)\|_2=1$, this derivative is the unit tangent
pointing in the direction of increasing $\theta$.
The signed tangential component of the increment
$g(r)=(1-\cos\theta,\sin\theta)$ is therefore
\begin{equation}
\begin{aligned}
g(r)\cdot r'(\theta)
&=(1-\cos\theta)\cos\theta-\sin^2\theta\\
&=\cos\theta-1
=-\frac{\theta^2}{2}+\O(\theta^4),
\end{aligned}
\label{eq:ogda-mm-tangent}
\end{equation}
where we used $\cos^2\theta+\sin^2\theta=1$ and the Taylor
expansion of $\cos\theta$ at zero.
The negative sign means that this component points toward decreasing
$\theta$, hence toward the equilibrium $r(0)=e=(0,1)$.
Its magnitude vanishes quadratically in $\theta$, explaining why
the drift toward equilibrium along the boundary becomes small.
Thus the motion toward equilibrium along the boundary is quadratic
in the current angular error. For the actual optimistic updates,
the auxiliary horizontal coordinate
$\widehat x_t=\Tr[X\widehat\alpha_t]$ satisfies
$\widehat x_{t+1}=\widehat x_t-(h/2)\widehat x_t^2
+\O(\widehat x_t^3)$. Its reciprocal consequently grows
asymptotically by $h/2$ per round, giving a $1/t$ rate of convergence in distance.
The proof derives this recurrence from both updates and also
computes the distinct rate of the duality gap.

\begin{theorem}[OGDA from maximally mixed states]
\label{thm:ogda-mm-rates}
The game~\eqref{eq:ogda-mm-game} has value zero and the unique
Nash equilibrium $\jointstar=(P,P)$. For every fixed
$0<\eta\le1/8$, the run~\eqref{eq:ogda-updates} initialized at
$\widehat\joint_0=\joint_0=(I/2,I/2)$ satisfies
\begin{equation}
\dist_F(\joint_t,\jointstar)\sim\frac{4}{\eta t} \qquad \text{and}
\qquad
G(\joint_t)\sim\frac{4}{\eta^3t^3}.
\label{eq:ogda-mm-rates}
\end{equation}
Every finite played iterate has positive distance and positive gap.
In particular, some constant $c_\eta>0$, independent of $t$ and
the target accuracy, satisfies
$\dist_F(\joint_t,\jointstar)\ge c_\eta/(t+1)$ for all $t\ge0$.
\end{theorem}

\begin{proof}
We identify the equilibrium, derive the exact symmetric trajectory,
and then calculate its asymptotic rate. Recall that $h=\eta/2$,
so $0<h\le1/16$.

\paragraph{Step 1: Identify the equilibrium.}
Maximizing~\eqref{eq:ogda-mm-payoff} over Bob's Bloch ball gives
\begin{equation}
\max_{\beta}f(\alpha,\beta)
=\frac14\left(\sqrt{x^2+(1-z)^2}-x\right)\ge0.
\label{eq:ogda-mm-alice-value}
\end{equation}
Equality requires $z=1$ and $x\ge0$. Feasibility with $z=1$
forces $\alpha=P$, so Alice's unique minimax strategy is $P$.
Since $f(P,\beta)=0$ for every $\beta$, the game value is zero.
Skew-symmetry, $f(\beta,\alpha)=-f(\alpha,\beta)$, similarly gives
\begin{equation}
\min_{\alpha}f(\alpha,\beta)
=\frac14\left(y-\sqrt{y^2+(1-w)^2}\right)\le0.
\label{eq:ogda-mm-bob-value}
\end{equation}
Equality forces $\beta=P$, so Bob's unique maximin strategy is
also $P$. Both equilibrium states are pure. Furthermore,
$A(P)=B(P)=0$, so both slack matrices vanish and strict
complementarity fails.

\paragraph{Step 2: Reduce the updates to a symmetric trajectory.}
We first establish the projection formula for an arbitrary Hermitian
input. Write $H=\tfrac12(sI+b_1X+b_2Y+b_3Z)$, where
$s\in\mathbb R$ and $b=(b_1,b_2,b_3)\in\mathbb R^3$.
For any $a=(a_1,a_2,a_3)$ with $\norm{a}_2\le1$, orthogonality
of the Pauli matrices gives
\begin{equation*}
\norm{\rho(a_1,a_2,a_3)-H}_F^2
=\frac12(1-s)^2+\frac12\norm{a-b}_2^2.
\end{equation*}
The first term is independent of $a$. Minimizing the second term
over the unit ball therefore gives
\begin{equation*}
\Pi_{\DD_2}(H)
=\rho\left(\frac{b_1}{L},\frac{b_2}{L},\frac{b_3}{L}\right),
\qquad L=\max\{1,\norm{b}_2\}.
\end{equation*}
This identity applies to every Hermitian input, regardless of its
trace or eigenvalues.

For $\rho=\rho(x,0,z)$ and
$\widehat\rho=\rho(\widehat x,0,\widehat z)$, the gradient formula
gives the exact identity
\begin{equation*}
\widehat\rho-\eta A(\rho)
=\frac12\left((1-hx)I
+\big(\widehat x+h(1-z)\big)X
+\big(\widehat z+hx\big)Z\right).
\end{equation*}
Since $B(\rho)=-A(\rho)$, the two players have identical
projection inputs whenever their current and auxiliary states agree.
The preceding projection formula also preserves the zero
$Y$-coordinate. Induction from the maximally mixed initialization
therefore yields
\begin{equation}
\alpha_t=\beta_t=\rho(x_t,0,z_t),
\qquad
\widehat\alpha_t=\widehat\beta_t
=\rho(\widehat x_t,0,\widehat z_t).
\label{eq:ogda-mm-symmetry}
\end{equation}
Let $\mathcal B_2=\{(x,z)\in\mathbb R^2:x^2+z^2\le1\}$,
$r_t=(x_t,z_t)$, $\widehat r_t=(\widehat x_t,\widehat z_t)$,
and $g(x,z)=(1-z,x)$. With
$\Pi_{\mathcal B_2}(r)=r/\max\{1,\norm{r}_2\}$, both updates
are exactly
\begin{equation}
\begin{aligned}
r_{t+1}
=\Pi_{\mathcal B_2}\big(\widehat r_t+h g(r_t)\big), \qquad
\widehat r_{t+1}
=\Pi_{\mathcal B_2}\big(\widehat r_t+h g(r_{t+1})\big),
\end{aligned}
\label{eq:ogda-mm-disk-updates}
\end{equation}
with $r_0=\widehat r_0=(0,0).$ These recurrences follow from projection onto the full qubit state spaces $\DD_2$.

\paragraph{Step 3: Prove finite boundary entry and positivity.}
The first updates are $r_1=(h,0)$ and
$\widehat r_1=(h,h^2)$; neither projection is active because
$h\le1/16$. Feasibility gives $1-z_t\ge0$, and radial
projection preserves nonnegative coordinates. Induction shows
that all coordinates remain nonnegative and that
$x_t,\widehat x_t>0$ for every $t\ge1$. Moreover,
\begin{equation}
\norm{g(r_t)}_2^2=x_t^2+(1-z_t)^2
\le2(1-z_t)\le2.
\label{eq:ogda-mm-increment-bound}
\end{equation}
Thus every projection denominator in
\eqref{eq:ogda-mm-disk-updates} lies between $1$ and
$M:=1+\sqrt2h$.

Suppose that no auxiliary vector ever reached the unit circle.
All auxiliary projections would then be inactive, so
$\widehat x_t\ge h$ for every $t\ge1$. The played update
would give $x_{t+1}\ge\widehat x_t/M\ge h/M$, and hence
\begin{equation*}
\widehat z_{t+1}
=\widehat z_t+h x_{t+1}
\ge\widehat z_t+\frac{h^2}{M}
\qquad(t\ge1).
\end{equation*}
This contradicts $\widehat z_t\le1$. A finite auxiliary boundary
entry time $T$ therefore exists. Once an auxiliary vector has
norm one, adding either nonnegative increment produces a vector
of norm at least one. Both projections then return unit vectors.
Consequently, $\norm{\widehat r_t}_2=1$ for $t\ge T$ and
$\norm{r_t}_2=1$ for $t\ge T+1$.

\paragraph{Step 4: Establish convergence to the equilibrium.}
Set $e=(0,1)$ and $K=-g$, so $K(x,z)=(z-1,-x)$.
The operator $K$ is $1$-Lipschitz and satisfies
$\langle K(r)-K(s),r-s\rangle=0$ for all $r,s$.
Also, $K(e)=0$ and $\langle K(r),r-e\rangle=0$.
Define $E_t=\norm{\widehat r_t-e}_2^2
+\tfrac1{16}\norm{\widehat r_t-r_t}_2^2$.
The Euclidean optimistic mirror-descent inequality
\cite[Lemma~1]{WeiEtAl2021}, with $h\le1/16<1/8$, gives
\begin{equation}
E_{t+1}\le E_t-\frac{15}{16}
\left(\norm{\widehat r_{t+1}-r_{t+1}}_2^2
      +\norm{r_{t+1}-\widehat r_t}_2^2\right).
\label{eq:ogda-mm-energy}
\end{equation}
Here the operator term vanishes because
$\langle K(r_{t+1}),r_{t+1}-e\rangle=0$.
Since $E_t\ge0$, the two squared differences in
\eqref{eq:ogda-mm-energy} are summable. In particular,
$\widehat r_t-r_t\to0$ and $r_{t+1}-\widehat r_t\to0$,
which also imply $r_{t+1}-r_t\to0$.

Let $r$ be a cluster point of $(r_t)$, and choose
$t_j\to\infty$ with $r_{t_j}\to r$. The preceding differences
give $\widehat r_{t_j}\to r$ and $r_{t_j+1}\to r$.
Passing to the limit in the played update yields
$r=\Pi_{\mathcal B_2}(r+h g(r))$.
There is no such fixed point in the interior of $\mathcal B_2$:
an interior projection output equals its input, so an interior
fixed point would satisfy $g(r)=0$. This forces $r=e$, which
lies on the boundary, a contradiction.

At a boundary fixed point, the radial projection formula gives
$g(r)=\lambda r$ for some $\lambda\ge0$. Writing $r=(x,z)$
and taking the inner product with $r-e$ gives
\begin{equation*}
0=\langle g(r),r-e\rangle
=\lambda\langle r,r-e\rangle
=\lambda(1-z).
\end{equation*}
If $\lambda=0$, then $g(r)=0$ and $r=e$. If $\lambda>0$,
then $z=1$, and feasibility again forces $r=e$.
Thus every cluster point is $e$. Compactness now implies
$r_t,\widehat r_t\to e$.

\paragraph{Step 5: Derive the quadratic scalar recurrence.}
We first compare the horizontal coordinates without assuming
their rate of decay. The auxiliary numerator is at least
$\widehat x_{t-1}$ and its denominator is at most $M$, so
$\widehat x_t\ge\widehat x_{t-1}/M$ for every $t\ge1$.
For $t\ge T+1$, the played vector is on the upper unit semicircle.
Hence $1-z_t=x_t^2/(1+z_t)\le x_t$, and the played update gives
\begin{equation}
\frac{x_{t+1}}{\widehat x_t}
\le1+h\frac{x_t}{\widehat x_t}
\le1+hM\frac{x_t}{\widehat x_{t-1}}
\qquad(t\ge T+1).
\label{eq:ogda-mm-comparison}
\end{equation}
To make this bound explicit, set
$q_t=x_t/\widehat x_{t-1}$ for $t\ge T+1$ and $a=hM<1$.
Then $q_{t+1}\le1+a q_t$, so iteration gives
\begin{equation*}
q_{T+1+n}
\le a^n q_{T+1}+\sum_{j=0}^{n-1}a^j
\le q_{T+1}+\frac1{1-a}
\qquad(n\ge0),
\end{equation*}
where the sum is zero when $n=0$. Thus
$x_{t+1}=\O(\widehat x_t)$ and, using
$\widehat x_{t-1}\le M\widehat x_t$, also
$x_t=\O(\widehat x_t)$.

Let $u=\widehat x_t$. After boundary entry, the played state
$\alpha_t=\rho(x_t,0,z_t)$ is pure and hence has zero determinant:
$\det(\alpha_t)=(1-x_t^2-z_t^2)/4=0$.
Therefore $x_t^2+z_t^2=1$.
Since $z_t\ge0$, we can factor $1-z_t^2$ and divide by
$1+z_t\ge1$ to obtain
\begin{equation*}
1-z_t=\frac{x_t^2}{1+z_t}\le x_t^2.
\end{equation*}
The previously established estimate $x_t=\O(u)$ now implies
$1-z_t=\O(u^2)$.
Moreover, $g(r_t)=(1-z_t,x_t)$ satisfies
\begin{equation*}
\|g(r_t)\|_2^2
=(1-z_t)^2+x_t^2
=(1-z_t)^2+1-z_t^2
=2(1-z_t)
=\O(u^2).
\end{equation*}
Since $u>0$, it follows that $\|g(r_t)\|_2=\O(u)$. Let
$L_t=\norm{\widehat r_t+h g(r_t)}_2$ be the played projection
denominator after boundary entry. Since
$\norm{\widehat r_t}_2=1$ and both vectors have nonnegative
coordinates, $1\le L_t\le1+h\norm{g(r_t)}_2=1+\O(u)$.
Consequently,
\begin{equation*}
x_{t+1}-u
=\frac{u+h(1-z_t)}{L_t}-u
=\frac{h(1-z_t)-u(L_t-1)}{L_t}
=\O(u^2).
\end{equation*}
We have therefore proved
\begin{equation}
x_{t+1}=\widehat x_t+\O(\widehat x_t^2).
\label{eq:ogda-mm-played-comparison}
\end{equation}
All implicit constants below may depend on the fixed step size,
but not on $t$.

Now put $v=x_{t+1}=u+\O(u^2)$. For sufficiently large $t$,
the auxiliary and played vectors are on the upper unit semicircle,
so their vertical coordinates are $\sqrt{1-u^2}$ and
$\sqrt{1-v^2}$, respectively. The auxiliary update is exactly
\begin{equation}
\widehat x_{t+1}
=\frac{u+h(1-\sqrt{1-v^2})}
{\sqrt{\big(u+h(1-\sqrt{1-v^2})\big)^2
       +\big(\sqrt{1-u^2}+hv\big)^2}}.
\label{eq:ogda-mm-aux-exact}
\end{equation}
Write $\delta(v)=1-\sqrt{1-v^2}$ and let $N$ denote this
denominator. Since $u\to0$ and $v=\O(u)$, Taylor expansion gives
$\delta(v)=v^2/2+\O(v^4)$ and
$\sqrt{1-u^2}=1-u^2/2+\O(u^4)$. Expanding the two squares,
and using $u^2+(\sqrt{1-u^2})^2=1$, gives
\begin{align*}
N^2
=1+2h\big(u\delta(v)+\sqrt{1-u^2}\,v\big)
       +h^2\big(\delta(v)^2+v^2\big),
\end{align*}
and consequently,
\begin{align*}
N^2-(1+hv)^2
=2h u\delta(v)
  +2h\big(\sqrt{1-u^2}-1\big)v
  +h^2\delta(v)^2.
\end{align*}
The three terms on the last line are respectively
$\O(u^3)$, $\O(u^3)$, and $\O(u^4)$.
Moreover, $N\ge1$ and $v\ge0$, so
\begin{equation*}
N-(1+hv)
=\frac{N^2-(1+hv)^2}{N+1+hv}
=\O(u^3).
\end{equation*}
It follows that
$N^{-1}=1-hv+h^2v^2+\O(u^3)$, while the numerator is
$u+(h/2)v^2+\O(u^4)$. Multiplying these expansions yields
\begin{align*}
\widehat x_{t+1}
=\left(u+\frac h2v^2+\O(u^4)\right)
  \left(1-hv+h^2v^2+\O(u^3)\right)
=u-huv+\frac h2v^2+\O(u^3).
\end{align*}
Here the omitted products have order at least $u^3$ because
$v=\O(u)$. Finally, $v=u+\O(u^2)$ gives
$uv=u^2+\O(u^3)$ and $v^2=u^2+\O(u^3)$, so
\begin{equation}
\widehat x_{t+1}
=u-\frac h2u^2+\O(u^3)
=\widehat x_t-\frac h2\widehat x_t^2
 +\O(\widehat x_t^3).
\label{eq:ogda-mm-scalar}
\end{equation}
To extract the rate, write $u_t=\widehat x_t>0$. Since $u_t\to 0$,
the recurrence gives
\begin{align*}
\frac1{u_{t+1}}-\frac1{u_t}
=\frac{u_t-u_{t+1}}{u_tu_{t+1}}
=\frac{\frac h2u_t^2+\O(u_t^3)}
        {u_t^2\big(1-\frac h2u_t+\O(u_t^2)\big)}
=\frac h2+\O(u_t)
\longrightarrow\frac h2.
\end{align*}
For a fixed sufficiently large $t_0$, summing these increments gives
\begin{equation*}
\frac1{t u_t}
=\frac1{t u_{t_0}}
+\frac1t\sum_{s=t_0}^{t-1}
 \left(\frac1{u_{s+1}}-\frac1{u_s}\right)
\longrightarrow\frac h2.
\end{equation*}
The limit follows because the averages of a convergent sequence
have the same limit. Thus $t\widehat x_t\to2/h$.
Equation~\eqref{eq:ogda-mm-played-comparison} also gives
$x_{t+1}/\widehat x_t\to1$, and therefore
\begin{equation*}
(t+1)x_{t+1}
=\left(1+\frac1t\right)
  \big(t\widehat x_t\big)
  \frac{x_{t+1}}{\widehat x_t}
\longrightarrow\frac2h.
\end{equation*}
Equivalently, $t x_t\to2/h$.

\paragraph{Step 6: Compute the two errors.}
By~\eqref{eq:ogda-mm-symmetry} and the Frobenius geometry of the
Bloch representation, the joint distance is
$d_t:=\dist_F(\joint_t,\jointstar)
=\sqrt{x_t^2+(1-z_t)^2}$.
After boundary entry, $1-z_t=x_t^2/(1+z_t)=\O(x_t^2)$.
Hence $d_t\sim x_t\sim2/(ht)=4/(\eta t)$.

The best-response values~\eqref{eq:ogda-mm-alice-value} and
\eqref{eq:ogda-mm-bob-value} give
$G(\joint_t)=(d_t-x_t)/2$ at every iteration.
Once the states are pure, $d_t^2=2(1-z_t)$ and
$x_t=d_t\sqrt{1-d_t^2/4}$. Therefore
\begin{equation}
\begin{aligned}
G(\joint_t)
=\frac{d_t}{2}\left(1-\sqrt{1-d_t^2/4}\right)
=\frac{d_t^3}{8\left(1+\sqrt{1-d_t^2/4}\right)}
\sim\frac{d_t^3}{16}.
\end{aligned}
\label{eq:ogda-mm-gap-distance}
\end{equation}
Substituting $d_t\sim4/(\eta t)$ proves
$G(\joint_t)\sim4/(\eta^3t^3)$.

At $t=0$, the distance is $1$ and the gap is $1/2$.
For every $t\ge1$, positivity gives $x_t>0$, and feasibility
then forces $z_t<1$. Thus $d_t>x_t>0$ and
$G(\joint_t)=(d_t-x_t)/2>0$.
Finally, choose an integer $t_1\ge1$ such that
$(t+1)d_t\ge2/\eta$ for all $t\ge t_1$, which is possible
by the distance asymptotic. The constant
\begin{equation*}
c_\eta
=\min\left\{\frac2\eta,
             \min_{0\le t<t_1}(t+1)d_t\right\}>0
\end{equation*}
then satisfies $d_t\ge c_\eta/(t+1)$ for every $t\ge0$.
\end{proof}

\begin{corollary}[Lower bounds at every iteration from the maximally mixed start]
\label{cor:ogda-mm-finite-time}
For the run in Theorem~\ref{thm:ogda-mm-rates}, fix
$0<\eta\le1/8$. There exists a constant $c_\eta>0$, depending only
on $\eta$, such that every integer $t\ge0$ satisfies
\begin{align}
\dist_F(\joint_t,\jointstar)
\ge \frac{c_\eta}{t+1} \qquad \text{and} \qquad
G(\joint_t)
\ge \frac{c_\eta^3}{16(t+1)^3}.
\label{eq:ogda-mm-gap-lower}
\end{align}
The constant $c_\eta$ is independent of the iteration index and the target
accuracy.
\end{corollary}

\begin{proof}
The distance bound is the final assertion of
Theorem~\ref{thm:ogda-mm-rates}. To obtain the gap bound, write
$d_t=\dist_F(\joint_t,\jointstar)$. The symmetric trajectory
\eqref{eq:ogda-mm-symmetry} and the best-response formulas give
$d_t^2=x_t^2+(1-z_t)^2$ and $G(\joint_t)=(d_t-x_t)/2$ at every
iteration. Feasibility gives $x_t^2+z_t^2\le1$, and hence
$d_t^2\le2(1-z_t)$. Also, $d_t>0$ and $0\le x_t\le d_t$ along
the trajectory. Therefore
\begin{equation*}
G(\joint_t)
=\frac{d_t-x_t}{2}
=\frac{(1-z_t)^2}{2(d_t+x_t)}
\ge\frac{d_t^3}{16}.
\end{equation*}
Here the numerator is at least $d_t^4/4$, while the denominator
is at most $4d_t$. Substituting
$d_t\ge c_\eta/(t+1)$ proves~\eqref{eq:ogda-mm-gap-lower}.
This argument also applies before the iterates become pure.
\end{proof}

\begin{corollary}[Iteration complexity from the maximally mixed start]
\label{cor:ogda-mm-complexity}
For the run in Theorem~\ref{thm:ogda-mm-rates}, the hitting times
defined in~\eqref{eq:hitting-times} satisfy
\begin{equation}
\tau_{\rm dist}(\varepsilon;\joint_0)
\sim\frac{4}{\eta\varepsilon},
\qquad
\tau_{\rm gap}(\varepsilon;\joint_0)
\sim\frac{4^{1/3}}{\eta\varepsilon^{1/3}}
\quad\text{as }\varepsilon\downarrow0.
\label{eq:ogda-mm-complexity}
\end{equation}
\end{corollary}

\begin{proof}
Theorem~\ref{thm:ogda-mm-rates} bounds each error above and below
by arbitrarily small relative perturbations of its leading term
at every sufficiently large index. Every fixed finite prefix has
a positive minimum error, so the first successful index belongs
to this asymptotic range when $\varepsilon$ is sufficiently small.
Inverting the upper and lower bounds proves the two equivalences.
No monotonicity of the errors is required.
\end{proof}

Thus the $\Theta(\varepsilon^{-1})$ distance bound holds with the
game, step size, and maximally mixed initialization all fixed.
It rules out a geometric distance estimate for this single
trajectory. The squared distance has order $1/t^{2}$, while the
duality gap has order $1/t^{3}$; the corresponding gap complexity
is $\Theta(\varepsilon^{-1/3})$.

\section{Last-Iterate Lower Bounds for OMMWU}
\label{sec:ommwu}

We now give a fixed game and a fixed initialization for which the
OMMWU trajectory itself converges polynomially, even though the
equilibrium is unique and strictly complementary.

\subsection{Game construction}
\label{subsec:ommwu-construction}

Consider the observable
\begin{equation}
U_{\rm M}=\frac13\left[-Z\otimes I+I\otimes Z
+\frac12X\otimes(I-Z)\right].
\label{eq:observable}
\end{equation}
For an arbitrary pair, set $x=\Tr[X\alpha]$, $z=\Tr[Z\alpha]$,
and $y=\Tr[Z\beta]$. Then
\begin{align}
f(\alpha,\beta)&=\frac13\left[-z+y+\frac{x}{2}(1-y)\right],
\label{eq:payoff}\\
A(\beta)&=\frac13\left[-Z+yI+\frac{1-y}{2}X\right],
\label{eq:alice-gradient}\\
B(\alpha)&=\frac13\left[\left(-z+\frac{x}{2}\right)I
+\left(1-\frac{x}{2}\right)Z\right].
\label{eq:bob-gradient}
\end{align}
The run starts at $\joint_0=(I/2,I/2)$, with any fixed $\eta>0$.
Write $x_t=\Tr[X\alpha_t]$, $z_t=\Tr[Z\alpha_t]$,
$y_t=\Tr[Z\beta_t]$, $d_t=1-y_t$, and $\kappa=\eta/3$.
Thus $x_0=z_0=y_0=0$ and $d_0=1$.

\begin{lemma}[Normalization and strict complementarity]
\label{lem:equilibrium}
The observable in~\eqref{eq:observable} satisfies
$\norm{U_{\rm M}}_{\mathrm{op}}=(1+\sqrt2)/3<1$. The game has the unique Nash
equilibrium $\jointstar=(P,P)$ and value zero. It is strictly
complementary: both equilibrium slack matrices have rank one and
positive eigenvalue $2/3$.
\end{lemma}

\begin{proof}
On the subspaces corresponding to Bob's $Z$-eigenvalues $+1$ and $-1$,
the operator $3U_{\rm M}$ acts on Alice's register as $I-Z$ and $-I-Z+X$,
respectively. Its eigenvalues are therefore $0$, $2$, and
$-1\pm\sqrt2$, proving the norm formula.

For every $\alpha$, the coefficient of $y$ in~\eqref{eq:payoff} is
$(1-x/2)/3\geq1/6$. Bob consequently has the unique best response
$\beta=P$: the maximum requires $y=1$, and a density matrix with
$\Tr[Z\beta]=1$ must equal $P$. Against $P$, Alice minimizes
$f(\alpha,P)=(1-z)/3$, whose unique minimizer is $\alpha=P$.
Thus $(P,P)$ is the unique equilibrium, with value zero.

At $(P,P)$, the gradients in
\eqref{eq:alice-gradient}--\eqref{eq:bob-gradient} give $v_*=0$ and
$W_A=W_B=(I-Z)/3$. Both slacks have eigenvalues $0$ and $2/3$ and
kernel equal to the range of $P$. Each equilibrium state has rank one,
so both strict-complementarity rank equalities hold.
\end{proof}

The distinction between strict complementarity and equilibrium
isolation already appears in
\cite[Example~3.6]{ITTV25}, where a continuum of pure equilibria
satisfies strict complementarity.
Lemma~\ref{lem:equilibrium} establishes both uniqueness and strict complementarity
for the present zero-sum game.
Theorem~\ref{thm:rates} nevertheless shows polynomial convergence of OMMWU.
Thus, even these two equilibrium properties together do not ensure
geometric convergence of the learning dynamics.

\subsubsection{Why the construction delays convergence}
\label{subsec:ommwu-mechanism}

The construction makes Bob approach $P$ exponentially fast while
leaving a persistent off-diagonal term in Alice's accumulated update.
This term causes the projector onto Alice's leading eigenspace to approach \(P\) only polynomially. 
We explain both parts of this
mechanism and the hyperbolic functions that describe the iterates.

\paragraph{The coupling leaves a nonzero off-diagonal term.}
The uncoupled terms $-Z\otimes I$ and $I\otimes Z$ in $U_\mathrm{M}$ contribute
$-\Tr[Z\alpha]/3$ and $\Tr[Z\beta]/3$ to the payoff, respectively.
Since Alice minimizes the payoff and Bob maximizes it, both terms
favor $P=(I+Z)/2$, the state with a population in $Z$ equal to one.
The coupling term $\tfrac12X\otimes(I-Z)$ contributes $(d_t/6)X$ to Alice's gradient, where
$d_t=\Tr[(I-Z)\beta_t]=1-y_t\ge0$.
This contribution vanishes exactly when Bob plays $P$.
Before then, Alice's descent step subtracts a positive multiple
of $X$ from the matrix inside her exponential update.

To describe how these contributions accumulate, we represent each
iterate as $\Lambda(H)=e^H/\Tr[e^H]$ and call $H$ a
\emph{score matrix}. This representation is unchanged by adding
a scalar multiple of $I$ (i.e., $\Lambda(H+cI)=\Lambda(H)$).
Moreover, $\log\Lambda(H)=H-\log\Tr[e^H]I$, so an update of the
form $\Lambda(\log\Lambda(H)+G)$ equals $\Lambda(H+G)$.
Thus, gradient increments add directly to the score matrix,
and scalar identity terms can be discarded.

With $\kappa=\eta/3$ and a time $s$, removing scalar identity terms from
Alice's descent increment $-\eta A(\beta_s)$ and Bob's ascent
increment $\eta B(\alpha_s)$ leaves, respectively,
\begin{equation*}
\kappa Z-\frac{\kappa d_s}{2}X\qquad \mathrm{and}
\qquad
\kappa\left(1-\frac{x_s}{2}\right)Z.
\end{equation*}
Alice's score therefore accumulates both diagonal and off-diagonal
contributions, while Bob's score remains proportional to $Z$.

Both initial scores can be taken to be zero because
$\Lambda(0)=I/2$. At time $t-1$, each auxiliary score contains
the increments evaluated at $s=1,\ldots,t-1$.
To form the played iterate at time $t$, the algorithm adds one
further increment evaluated at $t-1$; this additional term is
the optimistic correction.
Consequently, as justified in the proof of
Theorem~\ref{thm:rates}, we have
$\alpha_t=\Lambda(H_t)$ and $\beta_t=\Lambda(q_tZ)$ for $t\ge1$,
where
\begin{equation}
\begin{aligned}
p_t=\frac{\kappa}{2}
\left(\sum_{s=1}^{t-1}d_s+d_{t-1}\right),\qquad
q_t=\kappa
\left[\sum_{s=1}^{t-1}\left(1-\frac{x_s}{2}\right)
      +\left(1-\frac{x_{t-1}}{2}\right)\right],
\end{aligned}
\label{eq:score-coefficients}
\end{equation}
and
\begin{equation}
H_t=\kappa tZ-p_tX
=\begin{pmatrix}
\kappa t&-p_t\\
-p_t&-\kappa t
\end{pmatrix},
\qquad
R_t=\sqrt{\kappa^2t^2+p_t^2}.
\label{eq:score-matrix}
\end{equation}
The eigenvalues of $H_t$ are $\pm R_t$.
Empty sums are zero, so $d_0=1$ and $x_0=0$ give
$p_1=\kappa/2$ and $q_1=\kappa$.
The coefficient $\kappa t$ consists of $t-1$ accumulated copies
of $\kappa Z$ and one additional copy from the optimistic step.
The coefficient $-p_t$ records Bob's past deviations from $P$,
together with the optimistic correction.
Thus, even when Bob's current deviation becomes small,
Alice's score retains the off-diagonal contributions accumulated
earlier.

\paragraph{Why hyperbolic functions appear.}
For a real linear combination $H=aZ+bX$, the Pauli identities give
$H^2=(a^2+b^2)I$. Let $R=\sqrt{a^2+b^2}>0$. The even powers of
$H$ are $R^{2j}I$, and the odd powers are $R^{2j}H$. The standard exponential identity for Pauli matrices then gives
\begin{equation}
e^H
=\left(\sum_{j=0}^{\infty}\frac{R^{2j}}{(2j)!}\right)I
 +\left(\sum_{j=0}^{\infty}\frac{R^{2j+1}}{(2j+1)!}\right)\frac{H}{R}
=\cosh(R)I+\frac{\sinh(R)}{R}H.
\label{eq:pauli-exponential}
\end{equation}
See \cite[Section 3.4.2]{Tisza2009} for related discussion. Here $\cosh R=(e^R+e^{-R})/2$ and
$\sinh R=(e^R-e^{-R})/2$. These functions combine the exponentials
of the two real eigenvalues $+R$ and $-R$. Since $\Tr[H]=0$, the
normalizing trace is $2\cosh R$. Dividing by it introduces
$\tanh R=\sinh R/\cosh R$ and gives
$\Lambda(H)=\tfrac12(I+\tanh(R)H/R)$. When $R=0$, we have $a=b=0$ and hence $H=0$, so the definition
of the normalized exponential gives $\Lambda(0)=I/\Tr[I]=I/2$.
This also agrees with the limit of the expression of $\Lambda(H)$: as $R\to0$,
the scalar factor $\tanh R/R$ tends to $1$, while $H\to0$.
Consequently, $\tanh(R)H/R\to0$ and $\Lambda(H)\to I/2$.
For our iterates, it follows that
\begin{equation}
\alpha_t=\frac12\left(I+\frac{\tanh R_t}{R_t}
                       (\kappa tZ-p_tX)\right),
\qquad
\beta_t=\frac12(I+\tanh(q_t)Z),
\label{eq:states-explicit}
\end{equation}
and hence
\begin{equation}
x_t=-\frac{p_t}{R_t}\tanh R_t,\qquad
z_t=\frac{\kappa t}{R_t}\tanh R_t,\qquad
y_t=\tanh q_t.
\label{eq:coordinates}
\end{equation}
Thus the hyperbolic functions follow directly from the matrix
exponential and its trace normalization.

\paragraph{Fast concentration and slow rotation.}
The coefficients above satisfy $p_t\ge0$, so $x_t\le0$ and
$q_t\ge\kappa t$. Consequently,
$d_t=1-\tanh q_t\le2e^{-2\kappa t}$ for $t\ge1$.
Bob's deviations are positive and summable. Equation~\eqref{eq:score-coefficients}
therefore gives
$p_t\to p_\infty=(\kappa/2)\sum_{s\ge1}d_s>0$:
the sum retains a nonzero contribution, whereas the additional
optimistic term $d_{t-1}$ tends to zero. Positivity follows already
from $d_1=1-\tanh\kappa>0$. The proof below also bounds the
convergence rate of $p_t$.

The score matrix $H_t$ has eigenvalues $\pm R_t$. Its \emph{leading
eigenvector} means an eigenvector for $+R_t$; its rank-one projector
is $P_t^+=(I+H_t/R_t)/2$. Exponentiation preserves these eigenvectors,
and the smaller eigenvalue of $\alpha_t$ is
$\lambda_t=(1+e^{2R_t})^{-1}\le e^{-2\kappa t}$. Thus
$\alpha_t=(1-\lambda_t)P_t^++\lambda_t(I-P_t^+)$ approaches the
moving pure state $P_t^+$ exponentially fast.

The direction of that pure state still changes slowly. Its Bloch
vector is $(-p_t,0,\kappa t)/R_t$, whereas $P$ has Bloch vector
$(0,0,1)$. Let $\theta_t\in(0,\pi/2)$ be the angle between them.
Then
\begin{equation}
\tan\theta_t=\frac{p_t}{\kappa t},\qquad
\theta_t\sim\frac{p_\infty}{\kappa t},\qquad
\|P_t^+-P\|_F=\sqrt{1-\frac{\kappa t}{R_t}}
\sim\frac{p_\infty}{\sqrt2\,\kappa t}.
\label{eq:ommwu-angle-mechanism}
\end{equation}
The angle is a Bloch-vector angle and a corresponding unit eigenvector
is $(\cos(\theta_t/2),-\sin(\theta_t/2))^{\mathsf T}$.
Alice's distance to $P$ consequently has order $1/t$. Her
population in $Q$ is of order $\theta_t^2$, which leads to a gap
of order $1/t^{2}$ once Bob's exponentially small contribution is
included. The relative entropy has order $1/t$ because its
leading term is
$R_t-\kappa t=p_t^2/(R_t+\kappa t)$.
Theorem~\ref{thm:rates} verifies these statements and their exact
constants for the played iterates.

\paragraph{A simple illustration.}
Fix the off-diagonal coefficient at one and consider
$K_s=\left(\begin{smallmatrix}s&-1\\-1&-s\end{smallmatrix}\right)$
for $s>0$. This family isolates the matrix effect in $H_t$.
The smaller eigenvalue of $\Lambda(K_s)$ is at most $e^{-2s}$,
but its off-diagonal entry is
$-\tanh(\sqrt{s^2+1})/(2\sqrt{s^2+1})\sim-1/(2s)$.
For large $s$, doubling the diagonal coefficient approximately
halves this entry. The normalized exponential is already close to
a pure state, whose direction still differs from the range of $P$.
Figure~\ref{fig:ommwu-rotation} depicts the same distinction along
the actual trajectory.

% Figure 2
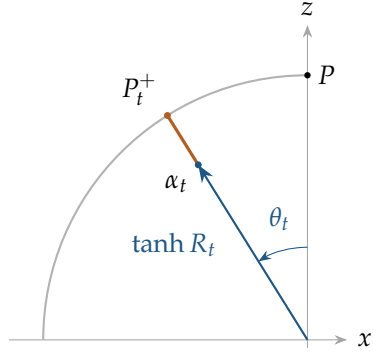
\begin{figure}[htbp]
\centering
\begin{tikzpicture}[x=3.5cm,y=3.5cm,>=Stealth,font=\small]
  \coordinate (O) at (0,0);
  \coordinate (N) at ({-sin(32)},{cos(32)});
  \coordinate (A) at ({-0.78*sin(32)},{0.78*cos(32)});

  \draw[->,gray!70] (-1.13,0) -- (0.15,0)
    node[right,text=black] {$x$};
  \draw[->,gray!70] (0,-0.03) -- (0,1.19)
    node[above,text=black] {$z$};

  \draw[gray!60,thick]
    (0,1) arc[start angle=90,end angle=180,radius=1];
  \draw[densely dashed,gray!70] (O) -- (N);
  \draw[->,constructionblue,thick] (O) -- (A);
  \draw[constructionorange,very thick] (A) -- (N);
  \draw[->,constructionblue] (0,0.35)
    arc[start angle=90,end angle=122,radius=0.35];
  \node[constructionblue] at (-0.105,0.46) {$\theta_t$};
  \node[constructionblue] at (-0.505,0.36) {$\tanh{R_t}$};

  \fill (0,1) circle[radius=0.012]
    node[right] {$P$};
  \fill[constructionorange] (N) circle[radius=0.013]
    node[above left,text=black] {$P_t^+$};
  \fill[constructionblue] (A) circle[radius=0.013]
    node[below left,text=black] {$\alpha_t$};
\end{tikzpicture}
\caption{OMMWU in the $(x,z)$ cross-section of the Bloch ball.
Alice's state $\alpha_t$ has Bloch radius $\tanh R_t$ and lies on
the radius through the pure state $P_t^+$ associated with the
leading eigenvector of $H_t$. The orange segment represents the
radial concentration $1-\tanh R_t=\O(e^{-2\kappa t})$.
The angle between this radius and the equilibrium direction satisfies
$\theta_t\sim p_\infty/(\kappa t)$. Thus Alice's state approaches
the moving pure state $P_t^+$ exponentially fast, while its direction
approaches $P$ at rate $1/t$. The radial separation is enlarged
to distinguish the two errors.}
\label{fig:ommwu-rotation}
\end{figure}
\FloatBarrier

\subsection{Proof of lower bounds}
In this section, we prove the lower bound of iteration complexity for OMMWU on the game instance we constructed in Section \ref{subsec:ommwu-construction}.
\subsubsection{Exact iterates and asymptotic rates}
\label{subsec:rates}

\begin{theorem}[Polynomial last-iterate convergence]
\label{thm:rates}
Run~\eqref{eq:played-update}--\eqref{eq:aux-update} on $U_{\rm M}$
from $\widehat\joint_0=\joint_0=(I/2,I/2)$, with any fixed
$\eta>0$. With $\kappa=\eta/3$, there exists a constant
$p_\infty\in(0,\infty)$, depending only on $\eta$, such that
\begin{align}
\dist_F(\joint_t,\jointstar)
&\sim\frac{p_\infty}{\sqrt2\,\kappa} \cdot \frac{1}{t},
\label{eq:distance-limit}\\
S(\jointstar\|\joint_t)
&\sim\frac{p_\infty^2}{2\kappa} \cdot \frac{1}{t},
\label{eq:entropy-limit}\\
G(\joint_t)
&\sim\frac{p_\infty^2}{6\kappa^2} \cdot \frac{1}{t^2}.
\label{eq:gap-limit}
\end{align}
In particular, the distance and quantum relative entropy decay as
$\Theta(1/t)$, and the duality gap decays as $\Theta(1/t^{2})$.
Every asymptotic statement holds with the game, step size, and
initialization fixed.
\end{theorem}

% \begin{theorem}[Polynomial last-iterate convergence]
% \label{thm:rates}
% Run~\eqref{eq:played-update}--\eqref{eq:aux-update} on $U_{\rm M}$
% from $\joint_0=(I/2,I/2)$, with any fixed $\eta>0$.
% There exists a constant
% $p_\infty\in(0,\infty)$, depending on $\eta$, such that
% \begin{align}
% \lim_{t\to\infty}t\,\dist_F(\joint_t,\jointstar)
%  &=\frac{p_\infty}{\sqrt2\,\kappa},
% \label{eq:distance-limit}\\
% \lim_{t\to\infty}t^2G(\joint_t)
%  &=\frac{p_\infty^2}{6\kappa^2},
% \label{eq:gap-limit}\\
% \lim_{t\to\infty}t\,S(\jointstar\|\joint_t)
%  &=\frac{p_\infty^2}{2\kappa}.
% \label{eq:entropy-limit}
% \end{align}
% In particular, the three errors are respectively
% $\Theta(1/t)$, $\Theta(1/t^{2})$, and $\Theta(1/t)$. Every
% asymptotic statement holds with the game and the step size fixed.
% \end{theorem}

\begin{proof}
We first derive the iterates exactly, then compute each error.

\paragraph{Step 1: Unroll the updates, including the optimistic correction.}
The identities for $\Lambda$ from Section~\ref{sec:preliminaries},
applied inductively to~\eqref{eq:aux-update}, give
\begin{equation}
\widehat\alpha_t=\Lambda\left(-\eta\sum_{s=1}^t A(\beta_s)\right),
\qquad
\widehat\beta_t=\Lambda\left(\eta\sum_{s=1}^t B(\alpha_s)\right).
\label{eq:aux-unrolled}
\end{equation}
The formula also holds at $t=0$ with empty sums, since
$\Lambda(0)=I/2$. Substituting into~\eqref{eq:played-update}, for
$t\geq1$ we obtain
\begin{align}
\alpha_t&=\Lambda\left(-\eta\Big[
             \sum_{s=1}^{t-1}A(\beta_s)+A(\beta_{t-1})\Big]\right),
\label{eq:alice-unrolled}\\
\beta_t&=\Lambda\left(\eta\Big[
             \sum_{s=1}^{t-1}B(\alpha_s)+B(\alpha_{t-1})\Big]\right).
\label{eq:bob-unrolled}
\end{align}
The last term in each bracket is the optimistic correction. These
identities use no commutativity assumption; in particular, they do not
replace the exponential of a sum by a product of exponentials.

Substituting \eqref{eq:alice-gradient}--\eqref{eq:bob-gradient}
and discarding scalar identity terms gives
\begin{equation}
\alpha_t=\Lambda(H_t)=\Lambda(\kappa tZ-p_tX),\qquad
\beta_t=\Lambda(q_tZ).
\label{eq:exact-iterates}
\end{equation}
The coefficients are exactly those in \eqref{eq:score-coefficients}.
In particular, $p_1=\kappa/2$ and $q_1=\kappa$. The state space has
not been restricted to real or diagonal matrices: the displayed forms
follow from the updates themselves.

\paragraph{Step 2: Prove that the off-diagonal coefficient has a positive limit.}
Applying the exponential identity~\eqref{eq:pauli-exponential} to
the exact scores~\eqref{eq:exact-iterates} proves the state and
coordinate formulas~\eqref{eq:states-explicit}--\eqref{eq:coordinates}.
We now use them to control $p_t$.

Every density matrix satisfies $y_s\leq1$, so $d_s\geq0$ and hence
$p_t\geq0$. Equation~\eqref{eq:coordinates} implies $x_t\leq0$.
Together with $x_0=0$, equation~\eqref{eq:score-coefficients} now yields
$q_t\geq\kappa t$. Thus
\begin{equation}
0<d_t=\frac{2}{e^{2q_t}+1}\leq2e^{-2\kappa t}
\quad\text{for every }t\geq1.
\label{eq:bob-exponential}
\end{equation}
The strict inequality follows because all matrices and coefficients
are finite at each finite iteration.

The series $\sum_{s\geq1}d_s$ consequently converges, and the
limiting coefficient is
\begin{equation}
p_\infty=\frac{\kappa}{2}\sum_{s=1}^{\infty}d_s,
\qquad
\frac{\kappa}{2}(1-\tanh\kappa)
\leq p_\infty\leq\frac{\kappa}{e^{2\kappa}-1}.
\label{eq:positive-limit}
\end{equation}
The lower bound uses $d_1=1-\tanh\kappa>0$, so $p_\infty$ is
strictly positive. Moreover,
\begin{equation}
p_t-p_\infty
=\frac{\kappa}{2}\left(d_{t-1}-\sum_{s=t}^{\infty}d_s\right)
=\O(e^{-2\kappa t})
\quad\text{as }t\to\infty.
\label{eq:p-convergence}
\end{equation}
The constant in the last bound may depend on $\kappa$. In particular,
the extra optimistic term tends to zero, while the accumulated
off-diagonal coefficient has the nonzero limit~\eqref{eq:positive-limit}.

\paragraph{Step 3: Compute the Frobenius distance.}
Since $p_t\to p_\infty$ and $R_t\geq\kappa t$, we have
$\tanh R_t=1+\O(e^{-2\kappa t})$ and
\begin{equation}
R_t=\kappa t+\frac{p_\infty^2}{2\kappa t}+o(1/t).
\label{eq:R-expansion}
\end{equation}
It follows from~\eqref{eq:coordinates} that
\begin{equation}
x_t=-\frac{p_\infty}{\kappa t}+o(1/t),
\qquad
1-z_t=\frac{p_\infty^2}{2\kappa^2t^2}+o(1/t^{2}).
\label{eq:coordinate-asymptotics}
\end{equation}
One can obtain the second expansion directly, without subtracting
approximate expressions, from
\begin{equation}
1-z_t
=\frac{p_t^2}{R_t(R_t+\kappa t)}
 +\frac{\kappa t}{R_t}(1-\tanh R_t).
\label{eq:z-error-exact}
\end{equation}
The matrices $X$ and $Z$ are orthogonal in the Hilbert--Schmidt inner
product and both have squared Frobenius norm $2$. Hence
\begin{equation}
\dist_F(\joint_t,\jointstar)^2
=\frac{x_t^2+(1-z_t)^2+d_t^2}{2}.
\label{eq:distance-exact}
\end{equation}
Combining~\eqref{eq:bob-exponential},
\eqref{eq:coordinate-asymptotics}, and~\eqref{eq:distance-exact}
proves~\eqref{eq:distance-limit}.

\paragraph{Step 4: Compute the duality gap.}
Bob's best response remains $P$, so
$\max_{\beta'}f(\alpha_t,\beta')=(1-z_t)/3$. Against $\beta_t$,
Alice minimizes the expectation of $-Z+(d_t/2)X$, whose eigenvalues
are $\pm\sqrt{1+d_t^2/4}$. The minimum expectation over all density
matrices is its smallest eigenvalue. Therefore
\begin{equation}
\min_{\alpha'}f(\alpha',\beta_t)
=\frac{y_t-\sqrt{1+d_t^2/4}}{3},\qquad
G(\joint_t)=\frac13\left(
1-z_t+d_t+\sqrt{1+d_t^2/4}-1\right).
\label{eq:gap-exact}
\end{equation}
The term $d_t$ is exponentially small, and
\begin{equation*}
0\le\sqrt{1+d_t^2/4}-1
=\frac{d_t^2}{4(\sqrt{1+d_t^2/4}+1)}\le\frac{d_t^2}{8}.
\end{equation*}
Equation~\eqref{eq:coordinate-asymptotics} now
proves~\eqref{eq:gap-limit}.

\paragraph{Step 5: Compute the quantum relative entropy.}
Since $P$ is pure, $\Tr[P\log P]=0$. Also
$\Tr[PZ]=1$ and $\Tr[PX]=0$. Taking the logarithm in
\eqref{eq:exact-iterates} and using $\Tr[e^{H_t}]=2\cosh R_t$ gives
\begin{align}
S(P\|\alpha_t)
 &=\log(2\cosh R_t)-\kappa t
   =\frac{p_t^2}{R_t+\kappa t}+\log(1+e^{-2R_t}),
\label{eq:alice-entropy}\\
S(P\|\beta_t)
 &=\log(2\cosh q_t)-q_t
   =\log(1+e^{-2q_t}).
\label{eq:bob-entropy}
\end{align}
The logarithmic remainders are at most $e^{-2\kappa t}$ each.
Since $p_t\to p_\infty$ and $(R_t+\kappa t)/t\to2\kappa$,
summing~\eqref{eq:alice-entropy} and~\eqref{eq:bob-entropy}
proves~\eqref{eq:entropy-limit}.
\end{proof}

\subsubsection{Explicit lower bounds and iteration complexity}
\label{subsec:complexity}

The preceding limits imply asymptotic lower bounds. To obtain bounds
valid at every positive iteration, define
\begin{equation}
p_- =\frac{\kappa}{2}(1-\tanh\kappa),\qquad
p_+ =\frac{\kappa}{2}+\frac{\kappa}{e^{2\kappa}-1},\qquad
C_R=\sqrt{\kappa^2+p_+^2}.
\label{eq:finite-constants}
\end{equation}
The next corollary uses these constants to control the first hitting
times without assuming that the errors are monotone.

\begin{corollary}[Lower bounds at every iteration]
\label{cor:finite-time}
For the run in Theorem~\ref{thm:rates}, the constants in
\eqref{eq:finite-constants} satisfy $p_->0$, and every integer $t\ge1$
satisfies
\begin{align}
\dist_F(\joint_t,\jointstar)
 &\geq\frac{p_-\tanh\kappa}{\sqrt2\,C_R}\cdot\frac1t,
\label{eq:distance-lower}\\
S(\jointstar\|\joint_t)
 &\geq\frac{p_-^2}{C_R+\kappa}\cdot\frac1t,
\label{eq:entropy-lower}\\
G(\joint_t)
 &\geq\frac{p_-^2}{3C_R(C_R+\kappa)}\cdot\frac1{t^2}.
\label{eq:gap-lower}
\end{align}
\end{corollary}

\begin{proof}
For $t=1$, $p_1=\kappa/2\geq p_-$. For $t\geq2$, the sum in
\eqref{eq:score-coefficients} contains $d_1$, so $p_t\geq\kappa d_1/2=p_-$.
Furthermore, $0\leq d_s\leq1$ for $s\geq0$, by initialization and
$q_s\geq\kappa s$. Together with~\eqref{eq:bob-exponential}, this
gives
\begin{equation}
p_t\leq\frac{\kappa}{2}
 \left(1+\sum_{s=1}^{\infty}d_s\right)\leq p_+,
\qquad
\kappa t\leq R_t\leq C_Rt.
\label{eq:finite-p-bounds}
\end{equation}
Equations~\eqref{eq:coordinates} and~\eqref{eq:distance-exact}
therefore imply
\begin{equation*}
\dist_F(\joint_t,\jointstar)\ge\frac{|x_t|}{\sqrt2}
\ge\frac{p_-\tanh\kappa}{\sqrt2\,C_Rt}.
\end{equation*}
All terms in~\eqref{eq:gap-exact} are nonnegative, and
\eqref{eq:z-error-exact} gives
\begin{equation*}
1-z_t\ge\frac{p_t^2}{R_t(R_t+\kappa t)}
\ge\frac{p_-^2}{C_R(C_R+\kappa)t^2}.
\end{equation*}
This proves~\eqref{eq:gap-lower}.
Finally, dropping the nonnegative logarithmic terms in
\eqref{eq:alice-entropy}--\eqref{eq:bob-entropy} yields
\begin{equation*}
S(\jointstar\|\joint_t)\ge\frac{p_t^2}{R_t+\kappa t}
\ge\frac{p_-^2}{(C_R+\kappa)t},
\end{equation*}
which completes the proof.
\end{proof}

\begin{corollary}[Failure of logarithmic iteration complexity]
\label{cor:complexity}
For every fixed $\eta>0$, the hitting times of the trajectory in
Theorem~\ref{thm:rates} satisfy, as $\varepsilon\downarrow0$,
\begin{equation}
\tau_{\rm dist}(\varepsilon;\joint_0)=\Theta(\varepsilon^{-1}),\qquad
\tau_{\rm gap}(\varepsilon;\joint_0)=\Theta(\varepsilon^{-1/2}),\qquad
\tau_S(\varepsilon;\joint_0)=\Theta(\varepsilon^{-1}).
\label{eq:complexity}
\end{equation}
None of the three errors $E_t$ in Theorem~\ref{thm:rates} admits
a geometric estimate $E_t\leq C(1-\gamma)^t$ for all
sufficiently large $t$, with $C>0$ and $\gamma\in(0,1)$, even when
these constants depend on the game and the step size.
\end{corollary}

\begin{proof}
The all-iteration bounds of Corollary~\ref{cor:finite-time} give the
corresponding lower bounds on the first iteration at which each error
is at most $\varepsilon$. All three errors are positive at $t=0$ as well:
the distance is $1$, the gap is $(1+\sqrt5/2)/3$, and the relative
entropy is $2\log2$. Thus including $t=0$ does not affect the
small-$\varepsilon$ bounds. The positive finite limits in
Theorem~\ref{thm:rates} give matching eventual upper bounds of the
form $E_t\leq C / t^{a}$, with $a=1,2,1$, respectively, and hence the
matching upper bounds on the hitting times. Finally,
$t^a(1-\gamma)^t\to0$ for every fixed $a>0$ and
$\gamma\in(0,1)$, contradicting the positive polynomial lower bounds
if a geometric estimate were to hold.
\end{proof}

\begin{remark}[The role of noncommutativity and the scope of the result]
The smaller eigenvalue of Alice's state is
$\lambda_{\min}(\alpha_t)=1/(1+e^{2R_t})\leq e^{-2\kappa t}$.
Nevertheless, in the fixed basis of~\eqref{eq:pauli},
$(\alpha_t)_{01}=-p_t\tanh R_t/(2R_t)
\sim-p_\infty/(2\kappa t)$. This entry measures part of the deviation
from $P$ and proves that eigenvalue concentration alone cannot yield
geometric convergence of the state. In fact, the two conditional
operators $I-Z$ and $-I-Z+X$ used in
Lemma~\ref{lem:equilibrium} have nonzero commutator $[-Z,X]$, so no
fixed basis on Alice's register diagonalizes them simultaneously.

The game is independent of both $\eta$ and the target accuracy, and
the argument holds for every fixed positive step size, including
arbitrarily small ones. The result concerns the unregularized dynamics
\eqref{eq:played-update}--\eqref{eq:aux-update} and the Nash target.
It gives $\Theta(\varepsilon^{-1/2})$ iterations for the duality gap,
even though distance and relative entropy require
$\Theta(\varepsilon^{-1})$ iterations. Consequently, this last-iterate
calculation does not establish an $\Omega(\varepsilon^{-1})$ gap
lower bound or a lower bound for the averaged output in
\cite{VasconcelosEtAl2025}.
\end{remark}

\FloatBarrier
\section{Experiments}
\label{sec:experiments}

\subsection{Numerical illustration of Theorem~\ref{thm:ogda-trajectory-lower-bound}}
\label{subsec:ogda-experiment}

To examine the delay in Theorem~\ref{thm:ogda-trajectory-lower-bound},
we run OGDA~\eqref{eq:ogda-updates} on the fixed game
$U_{\rm G}=Q\otimes Q$ with step size $\eta=1/8$ and
$r\in\{10^{-1},10^{-2},10^{-3},10^{-4}\}$.
Both the played and auxiliary pairs start at $(\rho_r,P)$, with
$\rho_r$ as in~\eqref{eq:ogda-initial-family}. Each run performs
$\lceil6/(\eta\sqrt r)\rceil$ updates without an accuracy-based
stopping test. We record the played iterates and compute
$G(\joint_t)=a_t$ and
$\dist_F(\joint_t,\ZZ^*)=\norm{\alpha_t-P}_F$.

\begin{figure}[!htb]
\centering
\includegraphics[width=0.9\linewidth]{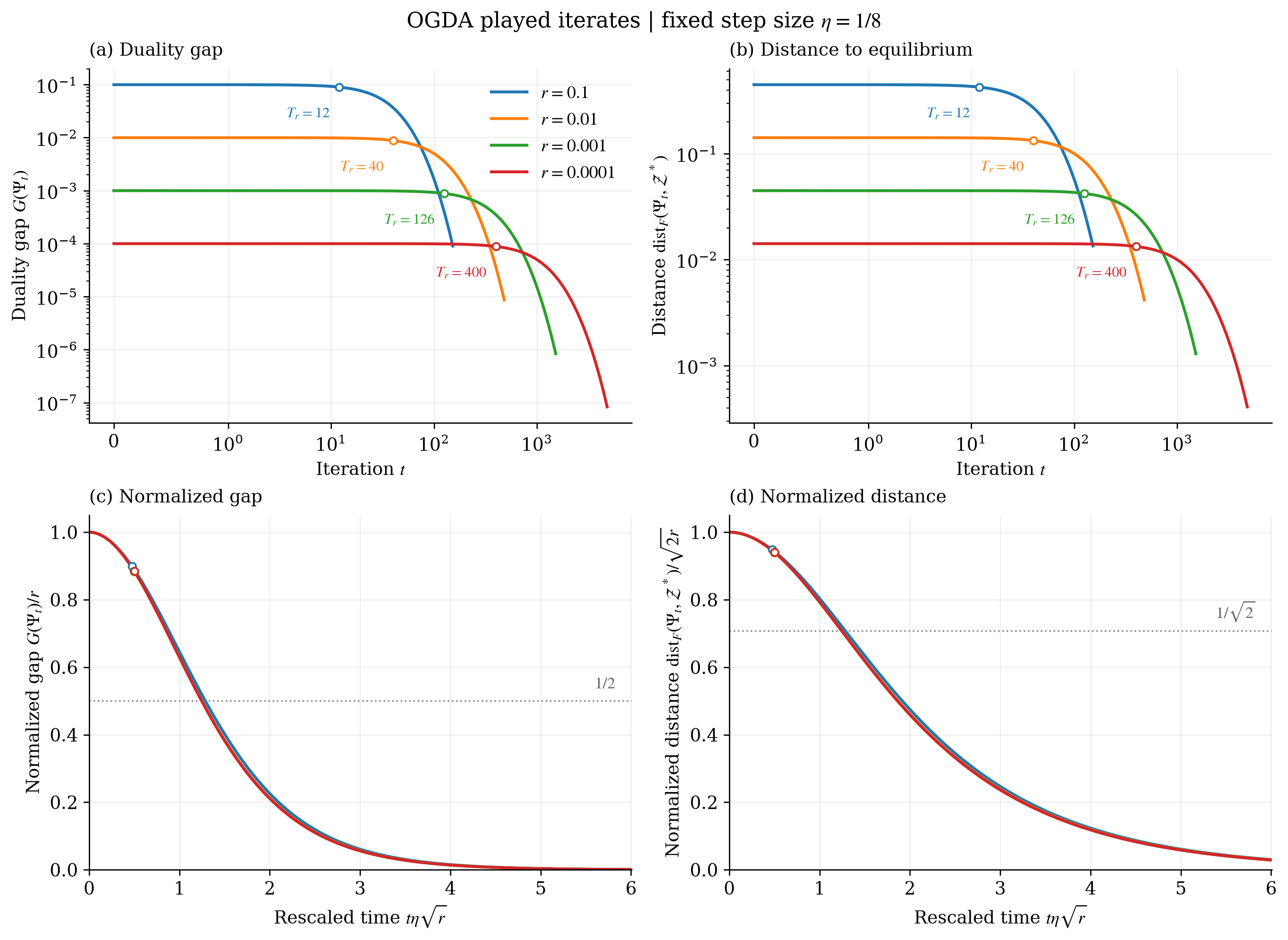}
\caption{OGDA on $U_{\rm G}=Q\otimes Q$ from $(\rho_r,P)$ with
$\eta=1/8$. Top: gap and distance versus iteration. Bottom: errors
normalized by their initial values versus $t\eta\sqrt r$.
Open circles mark $T_r$; dotted lines mark the normalized thresholds
$1/2$ and $1/\sqrt2$ from~\eqref{eq:ogda-delay-bound}. The top panels use logarithmic vertical axes and symlog horizontal axes to include \(t=0\).}
\label{fig:ogda-delay-experiment}
\end{figure}
\FloatBarrier

Figure~\ref{fig:ogda-delay-experiment}(a)--(b) shows increasingly
long initial plateaus as $r$ decreases, even though the initial
gap $r$ and distance $\sqrt{2r}$ become smaller.
For the four values of $r$ in decreasing order, the theorem horizons
are $T_r=12,40,126,400$. In every run, the gap remains above $r/2$
and the distance above $\sqrt r$ for all $0\le t\le T_r$,
in agreement with~\eqref{eq:ogda-delay-bound}.
The first iterations at which the gap is at most $r/2$ are
$33,100,316,998$, respectively, so the theorem's horizons are
conservative lower bounds on the observed delay.

After normalization and time rescaling, the curves nearly coincide
in Figure~\ref{fig:ogda-delay-experiment}(c)--(d).
The rescaled half-gap hitting times
$\eta\sqrt r\,\tau_{\rm gap}(r/2;(\rho_r,P))$ are approximately
$1.30,1.25,1.25,1.25$. These observations support a constant-factor
reduction time of order $1/(\eta\sqrt r)$ over the tested range,
consistent with Theorem~\ref{thm:ogda-trajectory-lower-bound}.
The later decrease is also compatible with the geometric bound
for each fixed initialization in
Proposition~\ref{prop:ogda-fixed-start-geometric}: the experiment
illustrates a delay that grows across initial states, without
asserting a polynomial asymptotic rate for any one of these runs.

\paragraph{Control from the maximally mixed state.}
To examine the effect of maximally mixed initialization, we run OGDA on the same game with the same step size, starting both pairs at \((I/2,I/2)\).
We keep the game $U_{\rm G}=Q\otimes Q$, step size $\eta=1/8$,
and OGDA updates unchanged, and add a run initialized at
$\widehat\joint_0=\joint_0=(I/2,I/2)$.
The control run performs $4800$ updates, including after either
error reaches zero. Figure~\ref{fig:ogda-initialization-control}
compares its first $32$ iterations with the same four coherent
starts $(\rho_r,P)$ used above. Each gap and distance is divided
by its own initial value, so the plots compare relative progress
on a common iteration axis.

\begin{figure}[!htb]
\centering
\includegraphics[width=\linewidth]{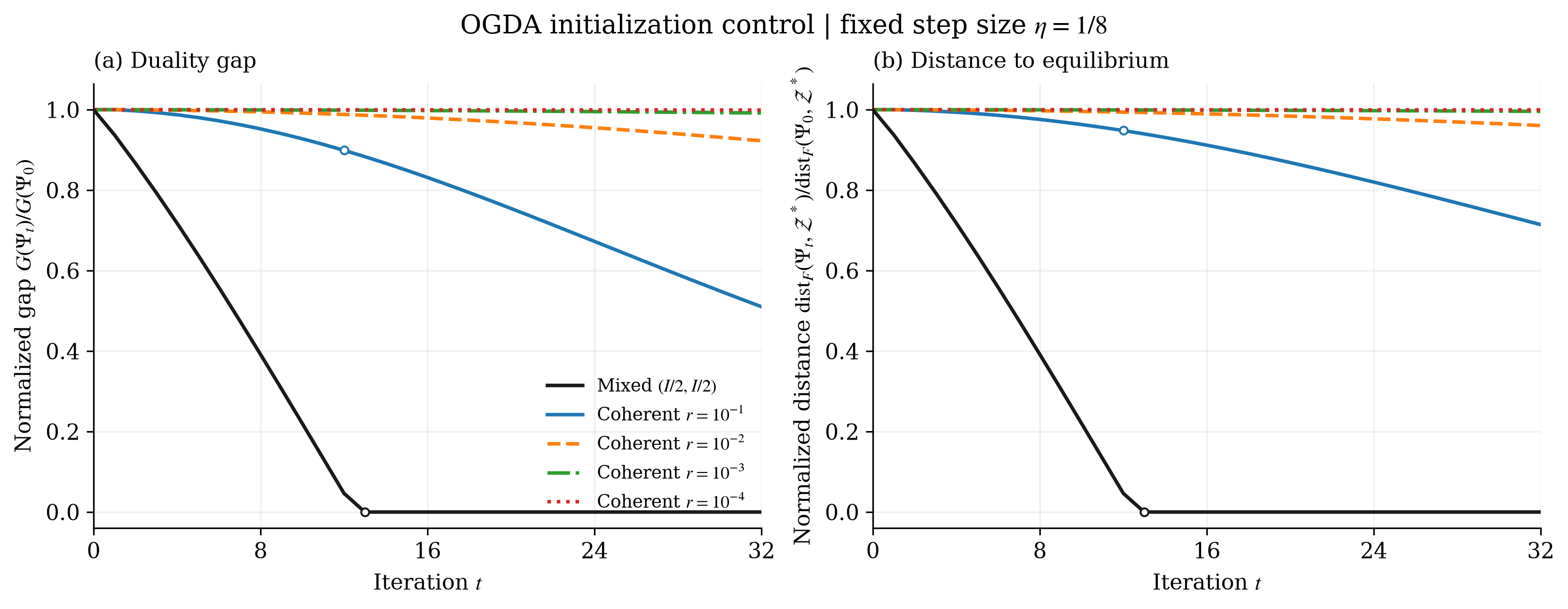}
\caption{Initialization control for OGDA on $U_{\rm G}=Q\otimes Q$
with $\eta=1/8$. Normalized gap (left) and distance (right) from
the maximally mixed state (black) and coherent starts (colored).
Open circles mark the mixed run's first zero at $t=13$ and the
coherent horizon $T_{10^{-1}}=12$; the other horizons exceed
the displayed window.}
\label{fig:ogda-initialization-control}
\end{figure}
\FloatBarrier

The mixed run first reaches $\alpha_{13}=P$, with both errors
equal to zero, and stays in the equilibrium set for the remainder
of the recorded trajectory. This agrees with
Proposition~\ref{prop:ogda-mixed-initialization}, which guarantees
termination by $\lceil2/\eta\rceil=16$.
Along this diagonal trajectory,
$\dist_F(\joint_t,\ZZ^*)=\sqrt2\,G(\joint_t)$, so the two
normalized black curves coincide. All four coherent runs retain
more than half their initial gap throughout the displayed window;
for smaller $r$, both normalized errors remain close to one.
The comparison illustrates the initialization dependence in
Theorem~\ref{thm:ogda-trajectory-lower-bound}: the growing delay
along the coherent family coexists with finite termination from
the maximally mixed state on the same game.

\subsection{Numerical illustration of Theorem~\ref{thm:ogda-mm-rates}}
\label{subsec:ogda-mm-experiment}

To illustrate Theorem~\ref{thm:ogda-mm-rates}, we run
OGDA~\eqref{eq:ogda-updates} on the fixed game~\eqref{eq:ogda-mm-game},
initializing both the played and auxiliary pairs at $(I/2,I/2)$.
We fix $\eta=1/8$ and perform $10^5$ updates without an
accuracy-based stopping test. At each played iterate, we measure
$\dist_F(\joint_t,\jointstar)$ and $G(\joint_t)$ relative to the
unique equilibrium $\jointstar=(P,P)$.
Both projections are implemented through the exact
Bloch-coordinate recurrence~\eqref{eq:ogda-mm-disk-updates},
including the initial phase inside the Bloch ball.
Writing $d_t=\sqrt{x_t^2+(1-z_t)^2}$ as above, we evaluate
the gap as $(1-z_t)^2/[2(d_t+x_t)]$ to avoid cancellation
in $(d_t-x_t)/2$. On the upper boundary, we also use
$1-z_t=x_t^2/(1+z_t)$. These are exact algebraic identities;
the asymptotic formulas are used only for comparison.

\begin{figure}[!htb]
\centering
\includegraphics[width=\linewidth]{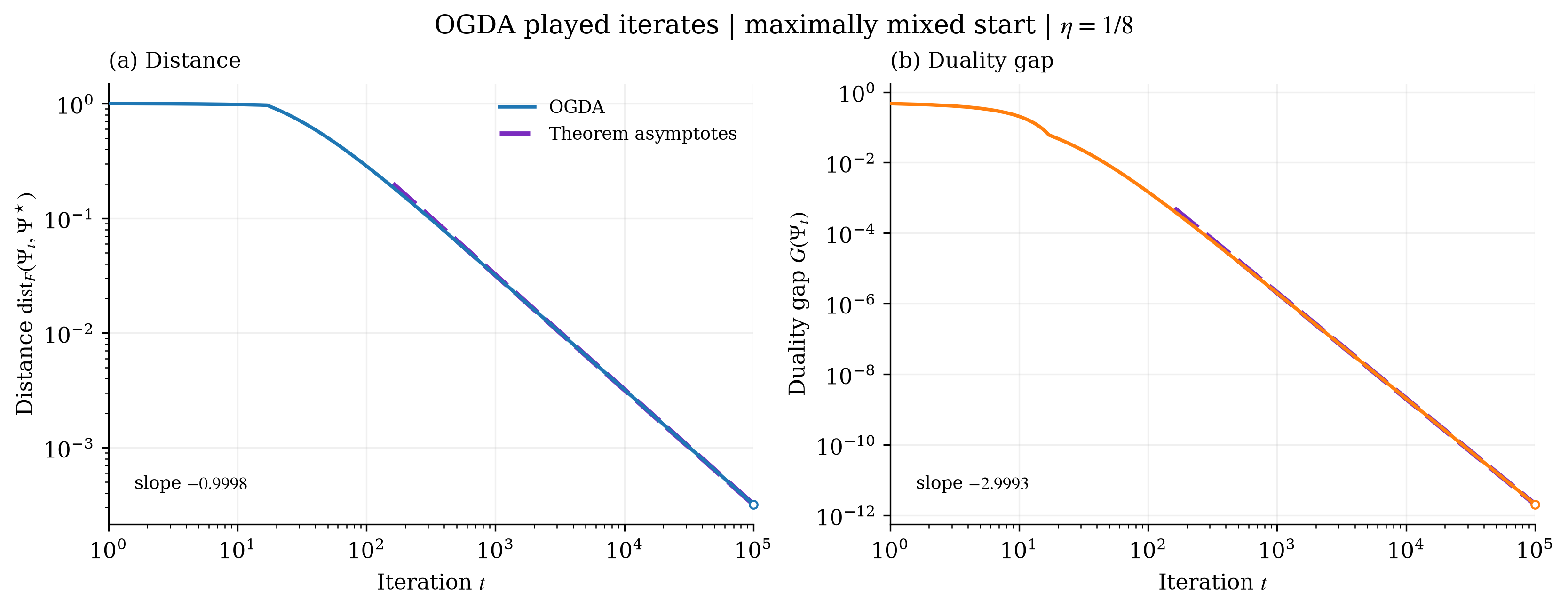}
\caption{OGDA from $(I/2,I/2)$ on $U_{\rm G}^{\rm mix}$ with
$\eta=1/8$. Played-iterate distance (left) and gap (right) on logarithmic
axes, with dashed asymptotes $32/t$ and $2048/t^3$ from
Theorem~\ref{thm:ogda-mm-rates}. Annotated slopes are fitted over
$10^4\le t\le10^5$. Open circles mark the final recorded iterate.}
\label{fig:ogda-mm-polynomial-experiment}
\end{figure}

After an initial transient, the distance and gap curves in
Figure~\ref{fig:ogda-mm-polynomial-experiment}(a)--(b) follow
their predicted asymptotes. Least-squares fits in logarithmic
coordinates at $400$ logarithmically spaced integer times between
$10^4$ and $10^5$ give slopes $-0.9998$ and $-2.9993$,
respectively, consistent with the $1/t$ distance and $1/t^{3}$
gap decay. The dashed curves use the theorem's constants directly,
independently of these fits.

\begingroup
\setlength{\emergencystretch}{1em}
The final errors also agree with the leading constants.
At \(t=10^5\), the rescaled errors are
\[
t\,\operatorname{dist}_F(\Psi_t,\Psi^\star)\approx31.99873646,
\qquad
t^3G(\Psi_t)\approx2047.75742248.
\]Their relative discrepancies from the predicted limits \(32\) and \(2048\), computed before rounding, are approximately \(3.95\times10^{-5}\) and \(1.18\times10^{-4}\), respectively.
 Both errors remain positive throughout
the recorded trajectory. These observations support polynomial
convergence for a single fixed game and initialization, with
the duality gap decaying faster than the distance.
\par\endgroup
\FloatBarrier

\subsection{Numerical illustration of Theorem~\ref{thm:rates}}
\label{subsec:ommwu-experiment}

To illustrate the rates in Theorem~\ref{thm:rates}, we run
unregularized OMMWU~\eqref{eq:played-update}--\eqref{eq:aux-update}
on the fixed game~\eqref{eq:observable}, initializing both the played
and auxiliary pairs at $(I/2,I/2)$. We fix $\eta=1/8$ and perform
$10^5$ updates without an accuracy-based stopping test.
At each played iterate, we measure
$\dist_F(\joint_t,\jointstar)$, $G(\joint_t)$, and
$S(\jointstar\|\joint_t)$ relative to the unique equilibrium
$\jointstar=(P,P)$. Both updates are implemented in score
coordinates, and the exact error formulas above are evaluated
in algebraically equivalent forms that avoid numerical cancellation
as the states approach rank one.

\begin{figure}[!htb]
\centering
\includegraphics[width=\linewidth]{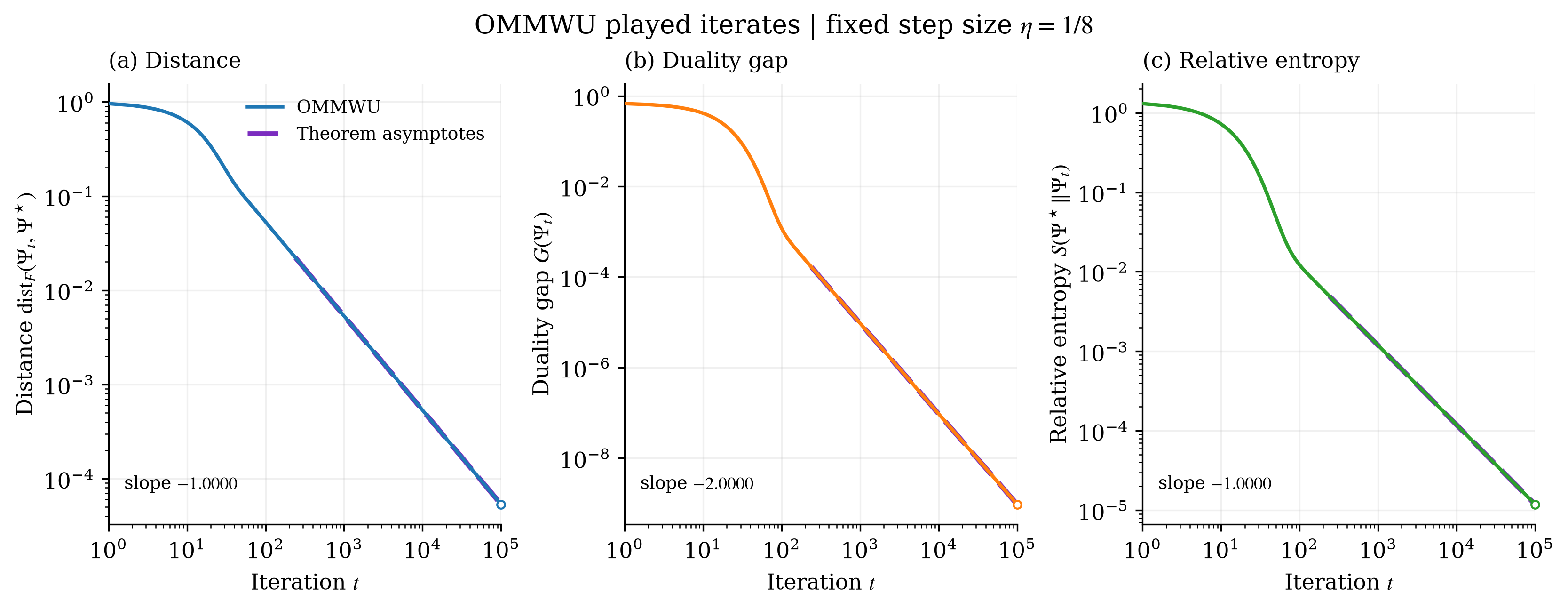}
\caption{OMMWU from $(I/2,I/2)$ on $U_{\rm M}$ with $\eta=1/8$.
Solid curves show the played-iterate errors on logarithmic axes;
dashed curves show the asymptotes from Theorem~\ref{thm:rates}.
Annotated slopes are fitted over $10^4\le t\le10^5$.}
\label{fig:ommwu-polynomial-experiment}
\end{figure}
\FloatBarrier

After an initial transient, all three curves in
Figure~\ref{fig:ommwu-polynomial-experiment} follow straight lines
on the logarithmic axes. Least-squares fits in logarithmic coordinates
at $400$ logarithmically spaced integer times between $10^4$ and
$10^5$ give slopes $-1$, $-2$, and $-1$, respectively,
to four decimal places. These agree with the $1/t$ distance,
$1/t^{2}$ gap, and $1/t$ relative-entropy decay in
Theorem~\ref{thm:rates}.

The reference curves also test the leading constants.
Truncating the series in~\eqref{eq:positive-limit} at the recorded
horizon gives $p_\infty\approx0.31455$. With \(\kappa=1/24\), the predicted limits of \(t\,\operatorname{dist}_F(\Psi_t,\Psi^*)\), \(t^2G(\Psi_t)\), and \(t\,S(\Psi^*\|\Psi_t)\) are approximately \(5.3381\), \(9.4985\), and \(1.1873\), respectively.
These series-based values determine the dashed curves independently
of the slope fits. At $t=10^5$, the rescaled errors
$t\,\dist_F(\joint_t,\jointstar)$, $t^2G(\joint_t)$, and
$t\,S(\jointstar\|\joint_t)$ agree with the unrounded reference
values to within $5\times10^{-9}$ in relative error.
Thus the experiment supports both the exponents and the leading
constants predicted for this single fixed trajectory, with distance
and relative entropy decaying more slowly than the duality gap.

\FloatBarrier

\section{Conclusion and Discussion}
\label{sec:conclusion-discussion}

We established explicit lower bounds for optimistic matrix mirror-prox
in quantum zero-sum games with one qubit per player. For the uniform
average that includes the maximally mixed initial iterate, our
construction gives an $\Omega(1/\varepsilon)$ iteration lower bound
for attaining duality gap at most $\varepsilon$, independently of the
regularizer and step size. For OGDA, we exhibited arbitrarily long
delays across initializations and a fixed trajectory from maximally
mixed states with distance $\Theta(1/t)$ and gap $\Theta(1/t^{3})$.
For OMMWU, we constructed a game with a unique, strictly complementary
equilibrium for which distance and quantum relative entropy decay as
$\Theta(1/t)$, while the gap decays as $\Theta(1/t^{2})$.
These results show that uniqueness and strict complementarity do not
ensure geometric convergence of unregularized OMMWU, and distinguish
the rates associated with different error measures and output rules.

All our constructions share a boundary feature: every Nash equilibrium
has at least one singular strategy. In particular, none admits a
\emph{full-rank Nash equilibrium}, meaning an equilibrium
$\jointstar=(\alpha^*,\beta^*)$ with
$\alpha^*\succ0$ and $\beta^*\succ0$.
Our lower bounds therefore do not rule out faster last-iterate
convergence in games admitting such an equilibrium, possibly under
additional regularity assumptions. Equivalently, this condition asks
that both players' sets of equilibrium strategies contain full-rank
states, which can be paired to form an equilibrium in a zero-sum game.
Understanding how convergence rates depend on equilibrium rank and
the smallest positive eigenvalues is a natural direction for further
study.

This limitation requires distinguishing equilibrium strategies from
best responses to a fixed opponent. Because the payoff is linear in
each player's state, a full-rank exact best response exists only when
the corresponding payoff operator is a scalar multiple of the
identity; in that case, every state is a best response.
Consequently, at any non-equilibrium profile, at least one player's
best-response set consists entirely of singular states. This property
holds along the non-equilibrium trajectories in our constructions,
but does not by itself characterize slow convergence.
Indeed, in Theorem~\ref{thm:ogda-mm-rates}, we have $A(P)=B(P)=0$, so both
best-response sets at $(P,P)$ equal $\DD_2$ and contain full-rank
states. Nevertheless, the unique equilibrium is the singular pair
$(P,P)$, and convergence is polynomial. Thus the existence of
full-rank elements in both best-response sets at equilibrium is
insufficient to exclude our lower bounds. The full-rank equilibrium
condition above is stronger.

Finally, faster last-iterate convergence should be distinguished from
faster convergence of a uniform average. Our average-iterate
construction isolates the effect of retaining the initial error with
weight $1/T$, which persists even if all subsequent iterates are
already equilibria. Alternative averaging rules and structural
conditions for fast last-iterate convergence therefore address
different limitations of the methods studied here.
\printbibliography
\end{document}

%% file: mancro.tex
\usepackage[utf8]{inputenc}
\usepackage[T1]{fontenc}
\usepackage[margin=1in]{geometry}
\usepackage{mathpazo}
\usepackage{amsmath,amsfonts,amssymb,amsthm,mathtools}
\usepackage{microtype}
\usepackage{xcolor}

\usepackage{graphicx}
\usepackage{tikz}
\usetikzlibrary{arrows.meta,calc}
\usepackage[font=small,labelfont=bf]{caption}
\usepackage{placeins}
\usepackage{enumitem}
\usepackage{etoolbox}
\definecolor{constructionblue}{RGB}{31,87,130}
\definecolor{constructionorange}{RGB}{178,94,35}

\definecolor{niceRed}{RGB}{190,38,38}
\definecolor{niceYellow}{HTML}{f5b400}
\definecolor{blueGrotto}{HTML}{059DC0}
\definecolor{royalBlue}{HTML}{057DCD}
\definecolor{navyBlue}{HTML}{0B579C}
\definecolor{yaleBlue}{HTML}{00356b}
\definecolor{limeGreen}{HTML}{81B622}
\definecolor{nicePurple}{HTML}{9c27b0}
\definecolor{lightRoyalBlue}{HTML}{def2ff}
\definecolor{gold}{HTML}{ffa300}

\usepackage{hyperref}
\hypersetup{
  colorlinks=true,
  urlcolor=blueGrotto,
  linkcolor=royalBlue,
  citecolor=navyBlue
}

\makeatletter
\renewcommand\paragraph{\@startsection{paragraph}{4}{\z@}%
  {\smallskipamount}%
  {-1em}%
  {\normalfont\normalsize\bfseries}}
\makeatother

\AtBeginDocument{%
  \AfterEndEnvironment{theorem}{\noindent\ignorespaces}%
  \AfterEndEnvironment{lemma}{\noindent\ignorespaces}%
  \AfterEndEnvironment{proposition}{\noindent\ignorespaces}%
  \AfterEndEnvironment{corollary}{\noindent\ignorespaces}%
  \AfterEndEnvironment{remark}{\noindent\ignorespaces}%
  \AfterEndEnvironment{proof}{\noindent\ignorespaces}%
}

\usepackage{soul}
\setul{2pt}{0.4pt}
\usepackage[framemethod=TikZ]{mdframed}
\mdfsetup{
  backgroundcolor=royalBlue!0,
  roundcorner=4pt,
  skipabove=5pt,
  skipbelow=5pt,
  linewidth=1pt
}

\usepackage[
  backref=true,
  backend=biber,
  natbib=true,
  style=alphabetic,
  sorting=alphabeticlabel,
  sortcites=true,
  minbibnames=3,
  maxbibnames=999,
  mincitenames=4,
  maxcitenames=4,
  minalphanames=4,
  maxalphanames=4,
  doi=false,
  isbn=false
]{biblatex}
\DeclareSortingTemplate{alphabeticlabel}{
  \sort[final]{\field{labelalpha}}
  \sort{\field{year}}
  \sort{\field{title}}
}
\AtEveryBibitem{\clearname{editor}\clearlist{location}}

\AtBeginRefsection{\GenRefcontextData{sorting=ynt}}
\AtEveryCite{\localrefcontext[sorting=ynt]}

\usepackage[capitalise,noabbrev,nameinlink,sort]{cleveref}
\crefname{section}{Section}{Sections}
\crefname{theorem}{Theorem}{Theorems}
\crefname{lemma}{Lemma}{Lemmas}
\crefname{proposition}{Proposition}{Propositions}
\crefname{corollary}{Corollary}{Corollaries}
\crefname{remark}{Remark}{Remarks}
\crefname{equation}{Equation}{Equations}
\crefname{appendix}{Section}{Sections}
\crefname{figure}{Figure}{Figures}